\documentclass[twocolumn]{IEEEtran}
\normalsize
\usepackage{amsmath, amssymb}
\usepackage{amsthm}
\usepackage{mathtools}
\usepackage[small]{caption}
\usepackage{amsthm,multirow,color,amsfonts}
\usepackage{tabulary}
\usepackage{subfigure}
\usepackage{graphicx}
\usepackage{setspace}
\usepackage{enumerate}
\usepackage{comment}
\usepackage{cite}
\usepackage{bm}
\usepackage{bbm}
\usepackage{algorithm}
\usepackage[noend]{algpseudocode}
\makeatletter
\def\BState{\State\hskip-\ALG@thistlm}
\makeatother
	
\newcommand{\norm}[1]{\| #1 \|}                 
\newtheorem{theo}{Theorem}

\newtheorem{prop}{Proposition}

\newtheorem{cor}{Corollary}

\newtheorem{lem}{Lemma}

\DeclareMathOperator*{\argmin}{arg\,min}

\DeclareMathOperator*{\SINR}{{\mathrm{SINR}}}

\DeclareMathOperator*{\SNR}{\mathrm{SNR}}   

\DeclareMathOperator*{\MQAM}{\sf{MQAM}}

\DeclareMathOperator*{\Dpwl}{\mathcal{D}_{\text pwl}}

\newcommand{\eqdef}{\triangleq}

\newcommand{\zz}{z_{p_{k+1} p_k}(i,p)} 
\newcommand{\zzopt}{z_{p_{k+1} p_k}^{*}(i,p)}   
\newcommand{\zzoptrelaxed}{z_{p_{k+1} p_k}^{**}(i,p)} 
\newcommand{\Dopt}{D^{*}}
\newcommand{\Doptrelaxed}{D^{**}}

\newcommand{\derya}[1]{\textit{\textcolor{blue}{Derya: #1}}}

\IEEEaftertitletext{\vspace{-1\baselineskip}}

\title{Transmission Delay Minimization via Joint Power Control and Caching in Wireless HetNets\vspace{0cm}}
\author{Derya~Malak, F. Volkan Mutlu, Jinkun Zhang, and Edmund M. Yeh
\thanks{Malak is with the Electrical, Computer, and Systems Engineering, Rensselaer Polytechnic Institute, Troy, NY 12180 USA (email: malakd@rpi.edu).} 
\thanks{Mutlu, Zhang, and Yeh are with Electrical and Computer Engineering Department, Northeastern University, Boston, MA 02115 USA (e-mail: fvmutlu@ece.neu.edu, zhang.jinku@northeastern.edu, eyeh@ece.neu.edu).} 
\thanks{Manuscript last revised: {\today}.}
}

\begin{document}
\maketitle
\begin{abstract}
A fundamental challenge in wireless heterogeneous networks (HetNets) is
to effectively utilize the limited transmission and storage resources in
the presence of increasing deployment density and backhaul capacity
constraints.   To alleviate bottlenecks and reduce resource consumption,
we design optimal caching and power control algorithms for multi-hop
wireless HetNets.  We formulate a joint optimization framework to
minimize the average transmission delay as a function of the caching
variables and the signal-to-interference-plus-noise ratios (SINR) which
are determined by the transmission powers, while explicitly accounting
for backhaul connection costs and the power constraints.

Using convex relaxation and rounding, we obtain a reduced-complexity
formulation (RCF) of the joint optimization problem, which can provide a
constant factor approximation to the globally optimal solution.  We then
solve RCF in two ways: 1) alternating optimization of the power and
caching variables by leveraging biconvexity, and 2) joint optimization
of power control and caching.  We characterize the necessary (KKT)
conditions for an optimal solution to RCF, and use strict
quasi-convexity to show that the KKT points 
are Pareto optimal for RCF. We then devise a subgradient projection algorithm to jointly update the caching and power variables, and show that under appropriate conditions, the algorithm converges at a linear rate to the 
local minima of RCF, under general SINR conditions.  We support our analytical findings with results from extensive numerical experiments. 
\end{abstract}

\begin{IEEEkeywords}
Wireless cache, power, joint power-caching optimization, biconvex, alternating optimization, Pareto optimality.
\end{IEEEkeywords}

\section{Introduction}
\label{intro}
The energy and cost efficiencies of wireless heterogeneous networks (HetNets) incorporating macro cells (MCs) and small cells (SCs) are critical for meeting the performance requirements of 5G wireless networks \cite{Andrews2014}. Design of these HetNets entails the fundamental challenge of optimally utilizing both the bandwidth and storage resources of the network to reduce the download or transmission delay and the energy costs. With the increasing deployment density in wireless networks, the backhaul capacity becomes the bottleneck. It is well known that caching can alleviate this bottleneck by replacing the backhaul capacity with storage capacity at SCs \cite{Shanmugam2013}, i.e., moving content closer to the wireless edge. Caching improves the transmission delay performance by bringing the popular data items in SCs that are faster or computationally cheaper to access than MCs. 
To optimize resource usage in wireless HetNets, designing caching and power control policies and the interplay between caching and transmission decisions remains an open challenge. Enabling this will help control the interference and minimize the transmission delay costs in wireless HetNet topologies.

\subsection{Current State of the Art and Motivation}
\label{relatedwork}

Research to date on cost optimization in the context of caching has focused on different perspectives. On one hand, there have been attempts to devise replacement algorithms that aim to optimize the caching gain, which is the reduction in the expected total file downloading delay achieved by caching at intermediate nodes. Simple, elegant, adaptive, and distributed approaches determining how to populate caches in a variety of networking applications abound. These include Che's analytical approximation to compute the probability of an item being in a Least Recently Used (LRU) cache 
\cite{Che2002}, in the context of web caches \cite{breslau1999web}, and extension of Che's decoupling approach to provide a unified analysis of caching for different replacement policies in \cite{Martina2014}. A simple and ubiquitous algorithm for populating caches in peer-to-peer networking is path replication, i.e., once a request for an item reaches a cache, every downstream node receiving the response caches the item \cite{CohShe2002}. 
Various cache eviction policies devised for a single cache primarily concern the optimization of the cache hit rate that describes the frequency of finding the searched item in the cache, or the latency that describes how long it takes for the cache to return a desired item  
\cite{dan1990approximate,Che2002,Martina2014}.

For networks of caches, time-to-live (TTL) caching is a better alternative \cite{berger2014exact}, \cite{Martina2014}, where items stay in a cache for predetermined times and are evicted when the timers expire. 
An age-based-threshold policy where cache stores all contents requested more times than a threshold 
\cite{leconte2016placing} captures temporal popularity changes via the Poisson shot noise model (SNM), and maximizes the hit ratio \cite{traverso2013temporal}. Hence, SNM is compatible with the TTL caching \cite{leonardi2015least}. 
Traditional cache eviction policies \cite{Che2002,Martina2014,fricker2012versatile}, e.g., LRU, LFU, FIFO, RR, provide gain by making content available locally and compromise between hit rate and latency, 
and can be arbitrarily suboptimal in terms of the expected caching gain \cite{IoannidisYeh2016}. However, as devised in the landmark paper \cite{MaddahAli2013Journal}, novel coded caching approaches can provide a global gain that derives from jointly optimizing placement and delivery. Furthermore, geographic caching approaches that capture the spatial diversity of content, as in \cite{Blaszczyszyn2014}, \cite{Golrezaei2014TWC}, \cite{ji2015fundamental}, \cite{Malak2016twc}, help optimize the placement.

There is an extensive literature on physical layer aspects of caching in wireless networks design. For example, the gain offered by local caching and broadcasting is characterized in the landmark paper \cite{MaddahAli2013Journal}. 
Works also include the analysis of the scaling of the per-user throughput 
and collaboration distance \cite{Ji2014}, the wireless caching capacity region which is the closure of the set of all achievable caching traffic \cite{niesen2012caching}, as well as single-hop and device-to-device \cite{Ji2014},\cite{MaddahAli2013Journal},\cite{ji2015throughput},\cite{zhang2016efficient},\cite{Malak2016_D2DCaching}, and multi-hop caching networks \cite{jeon2015caching}, \cite{IoannidisYeh2017}.

Recently, information centric networking (ICN) architectures have put emphasis on the traffic engineering and caching problems \cite{IoannidisYeh2016},\cite{dehghan2019utility} to effectively use both bandwidth and storage for efficient content distribution  \cite{YehHoCuiBuLiuLeo2014}, and optimize the network performance  \cite{mahdian2017mindelay}. 
Alternatively, there have been works focusing on jointly optimizing the caching gain and resource usage, e.g., a decentralized SC caching optimization, i.e., femtocaching, to minimize the download delay \cite{Shanmugam2013}, distributed optimization of caching gain given routing \cite{IoannidisYeh2016}, 
minimizing the total cost incurred in storing and accessing objects by building the Steiner trees \cite{BaeRajSwa2008}, jointly optimizing caching and routing to provide latency guarantees \cite{li2018dr}, and minimizing delay by taking into account congestion \cite{dehghan2015complexity}, and elastic and inelastic traffic \cite{abedini2013content}. 
Existing strategies have also focused on separately optimizing the caching gain or the throughput \cite{ChenPapKoun2016}, and optimizing spatial throughput via scheduling 
\cite{KeeBlaMuh2016}. From a resource management perspective, it is not sufficient to exclusively optimize caching or throughput, or delay. 

There exist several pertinent power control algorithms to optimize the resource usage in wireless networks \cite{xiao2003utility}, \cite{gupta2019optimization}, \cite{hanly1995algorithm}, \cite{hanly1999power}, \cite{hanly1996capacity}, or maximize throughput under latency considerations \cite{LiaDimTas2017}. However, delay optimization in wireless links is challenging 
because of interference and congestion. There exist power-aware routing algorithms for packet forwarding to balance the traffic between high-quality links and less reliable links, such as \cite{DviCar2009}, \cite{ChaTas2000}, joint optimization of power control, routing, and congestion control \cite{XiYeh2008}, and joint optimization of radio and computational resources under latency and power constraints \cite{ChiHeXinCheWanLuStaAbd2006,BarSarLor2013,SarScuBar2015}, as well as delay-optimal computation task scheduling at the mobile edge \cite{LiuMaoZhaLet2016}, and the minimum delay routing algorithm \cite{Gallager1977}. In addition, fog optimization-based effective resource allocation schemes for wireless networks have been devised in \cite{OueCalStrBar2015} to achieve high power efficiency while keeping a very high Quality of Experience under latency constraints, and in \cite{yemini2019fog} to maximize the sum rate of cellular networks.
However, none of these approaches or research on ICN architectures has jointly designed traffic engineering and cache placement strategies to optimize network performance in view of traffic demands.

Several papers have studied complexity and optimization issues of cost minimization as an offline, centralized caching problem under restricted topologies \cite{Shanmugam2013},\cite{dehghan2015complexity},\cite{BorGupWal2010},\cite{fleischer2006tight},\cite{CohShe2002},\cite{BaeRajSwa2008}. Despite the advent of different caching solutions, to the best of our knowledge, none of the above protocols focuses on the joint optimization of caching and power allocation or provides algorithmic performance guarantees in terms of the achievable costs via caching. 
Although most of these strategies suggest that intermediate caching can alleviate the average download delays, it is hard to quantify how this delay is affected by the resource allocation strategy in a HetNet setting. In this paper, we focus on jointly optimizing the network level performance in terms of transmission delay and caching, which can be increasingly skewed away from a strategy that places the items without accounting for the transmission delay\footnote{In this paper, we primarily consider the transmission delay assuming a lightly loaded system which we detail in Sect. \ref{systemmodel}.}.

\subsection{Methodology and Contributions}
\label{contributions}
In this paper, we study jointly optimal caching and power control for arbitrary multi-hop wireless HetNet topologies with nodes that have caching capabilities. Note that as the networks are becoming increasingly heterogeneous, MCs and SCs can co-exist in 5G, and all networks beyond it \cite{Andrews2014}. Dense SC deployment is the key for 5G networks to enhance the capacity, rendering a cost-efficient backhaul solution a key challenge.  

For a given caching HetNet topology with multi-hop transmissions\footnote{Routing is fixed and each request is a pair that is jointly determined by the item requested and the fixed multi-hop path traversed to serve this request.}, a set of finite cache storage capacities, a demand distribution on the content items known a priori, and a subset of nodes designated to store specific items, we devise algorithms for jointly optimal caching and power control to minimize the average transmission delay cost, i.e., the average download delay, per request. 
While end-to-end delay in systems is due to several key sources, including transmission delay, propagation delay, processing delay and queuing delay, we are primarily interested in a lightly loaded regime for which congestion-dependent latency costs can be neglected, and in which the link lengths are much smaller than the propagation speed of the signal, 
and each node can sustain a high service rate relative to the average rate at which items are arriving to be serviced. Hence the transmission delay is the major delay component. 
To accurately determine the transmission delay, we explicitly account for the transmission power, backhaul costs, and wireless interference.

Finding the optimum placement of files is proven to be NP-complete \cite{Shanmugam2013}. Hence, jointly optimal power control and caching to minimize the transmission delay is also NP-complete. We emphasize that our joint optimization framework is significantly different from the traditional approach which maximizes the caching gain only. This approach has been widely studied in the literature, such as in \cite{Shanmugam2013,IoannidisYeh2016,IoanYeh2018tnet,IoannidisYeh2017} and their follow-up works, where the link costs are fixed. This assumption is only true when the links are granted orthogonal frequencies and do not interfere, and the transmission powers are fixed, which is not the case in HetNets. Furthermore, when link costs are deterministic, caching gain always improves with increasing link costs. This requires high transmission powers and violates the purpose of cost minimization. In other words, savings via intermediate caching do not inform us about the actual achievable delay-cost via caching. This justifies our proposed framework in Sect. \ref{delayminimization}, where we consider the minimum achievable cost via caching by taking into account the joint behavior of link costs under resource constraints. 

Our main technical contributions include the following:
\begin{itemize}
    \item {\bf A reduced-complexity formulation (RCF) to the joint optimization problem.} We provide a constant factor approximation to the minimum average transmission delay-cost $D^o(X,S)$ of serving a request via jointly optimizing binary caching variables $X$ and real valued transmission powers $S$. 
    Using convex relaxation techniques, 
    we obtain an RCF of the joint optimization problem, with cost function $D(Y,S)$ which is not jointly convex, where $Y$ denote the relaxed caching variables. 
    We then round $Y$ to obtain an integral solution within a constant factor from the optimal solution to $D^o(X,S)$. 
    
    \item {\bf Sufficient conditions for biconvexity of $D(Y,S)$.} We provide a sufficient condition for the convexity of RCF in the logarithm of powers which yields a biconvex RCF objective. 
    This 
    condition pertains 
    to the high SINR regime and does not hold for general SINR values. We jointly optimize RCF under the biconvexity condition to provide an alternating optimization solution to minimizing $D(Y,S)$. 
    
    \item {\bf Joint optimization framework.} We jointly optimize RCF under the general setting which is not jointly convex. We obtain the following results: {\bf a)} $D(Y,S)$ is strictly quasi-convex, {\bf b)} necessary conditions for optimality of $D(Y,S)$, {\bf c)} generalized necessary conditions for optimality of $D(Y,S)$ assuming strict convexity of $\mathcal{D}_S$, and {\bf d)} Pareto optimality of the solution to $D(Y,S)$. 

    \item {\bf Subgradient projection algorithm.} We provide a subgradient projection algorithm  
    which 
    is guaranteed to converge to a local minimum of the RCF. Due to the non-differentiability and non-convexity of the relaxed problem, 
    we propose a subgradient projection algorithm with a modified Polyak's step size. We also give a simple method to calculate the projection and show that the algorithm converges at a linear rate.
    
\end{itemize}

Organization of the rest of the paper is follows. 
In Sect. \ref{systemmodel} we detail the wireless HetNet topology where each node has caching capability and adjustable transmission powers. We establish a transmission delay model of serving a request using multiple hops where the transmission delay is a nonlinear function of the signal-to-interference-plus-noise ratios (SINR). 
In Sect. \ref{delayminimization} we detail the joint optimization of delay in power and caching variables. This section contains the main technical contributions which are the necessary and sufficient conditions for joint optimality, and algorithms to attain the optimal points with provable theoretical guarantees. 
In Sect. \ref{numericalresults}, we numerically verify our analytical findings.  In Sect. \ref{conc}, we conclude the paper by pointing out the use cases including mobile edge and fog computing.

\section{Wireless Caching Model}
\label{systemmodel} 
We consider a multi-hop wireless HetNet topology consisting of different types of nodes, e.g., small cells (SCs), macro cells (MCs), and users. The network serves content requests routed over different paths. To alleviate the impact of limited backhaul capacity, availability, and long-distance reach it is desired that the network serves the requests via the SCs and multi-hop transmissions. While each MC or SC might have a fiber connection to the backhaul network in 5G, multi-hop relaying\footnote{Since the transceiver is the major source of power consumption in a node and long distance transmission requires high power, in some cases multi-hop routing can be more energy efficient than single-hop routing \cite{fedor2007problem,pevsovic2010single}.} is essential due to radio range limitations. However, increasing the number of hops arbitrarily may lead to an additional energy consumption incurred by relays. As a result, long-hop routing 
is a competitive strategy for many networks \cite{haenggi2005routing}. Furthermore, from a cost-effective perspective, each MC should allocate its resources to a smaller number of users, which balances the traffic between SCs and MCs \cite{Andrews2014}. 
We represent the network as a directed graph $\mathcal{G}(V, E)$ where $V$ is the collection of nodes such that a node $v\in V$ is either an MC, an SC or a user. We assume that all nodes $V$ transmit on the same frequency\footnote{If subsets of nodes are allocated different frequencies, as in OFDMA-based networks, then we can determine the resulting subset of interfering nodes \cite{dhillon2012modeling}. This also reduces the interference and improves the coverage performance.}, i.e., all transmissions interfere with each other. In $\mathcal{G}$, $E$ is the set of edges, where given $v,\ u \in V$, the edge $(v,u) \in E$ denotes the transmission link from $v$ to $u$. 
In Fig. \ref{Backhaul}, we illustrate the network and possible multi-hop paths where the users request different items. We provide the notation for the proposed multi-hop wireless network model in Table \ref{table:tab1}.

\begin{figure}[t!]
\centering
\includegraphics[width=\columnwidth]{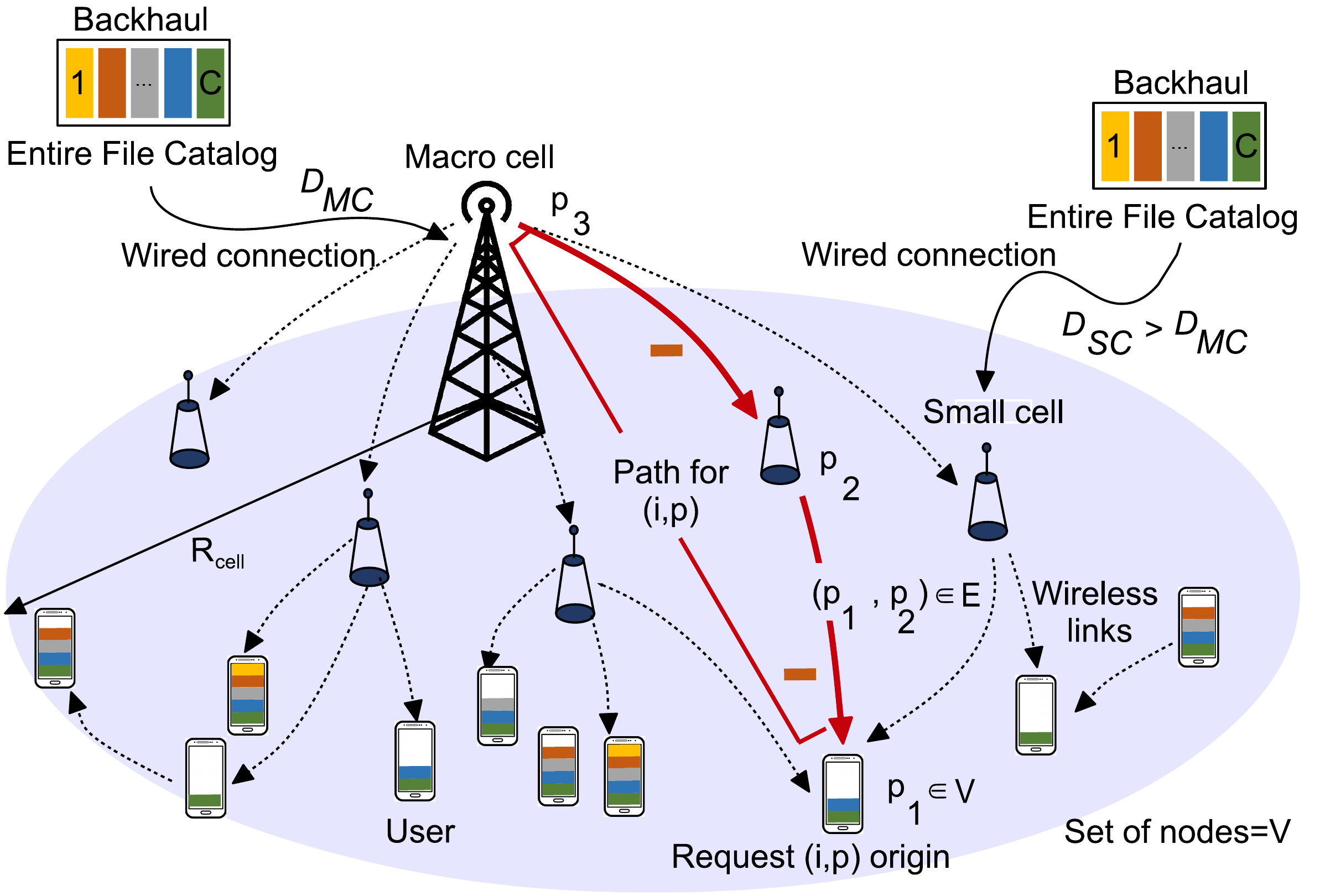}
\caption{A caching network scenario with possible connections between the users, SCs or MCs, and to the backhaul, where the backhaul cost $D_{SC}$ of SC is typically higher than the backhaul cost $D_{MC}$ of MC connections \cite{LiaDimTas2017}. A path $p=\{p_1,\,p_2,\,p_3\}$ for request $(i,p)$ is indicated where $p_1$ is a user where the request $(i,p)$ is originated, $p_2$ is a SC, and $p_3$ is the MC.}
\label{Backhaul}
\end{figure} 

The caching model is as follows. The entire set of content items, i.e., the catalog, is denoted by $\mathcal{C}$. Each item in $\mathcal{C}$ is of equal size. Each node is associated with a cache that can store a finite number of content items. The cache capacity at node $v\in V$ is $c_v$. The variables $x_{vi}\in\{0,1\}$ indicate whether $v\in V$ stores item $i\in \mathcal{C}$. Due to this finite capacity constraint, $\sum\nolimits_{i\in\mathcal{C}} {x_{vi}}\leq c_v$, $\forall v\in V$. Each item $i\in\mathcal{C}$ is associated with a fixed set of designated sources $\mathcal{S}_i\subseteq V$, i.e., nodes that always store $i$: $x_{vi}=1$, $\forall v\in\mathcal{S}_i$.  The designated sources could be user nodes, SCs or MCs. Items that are not available from the SCs need to be transmitted by the MCs with low-rate backhaul but high storage capacity.

Users issue requests for content items. The set of all requests is denoted by $\mathcal{R}$. A request $r\in\mathcal{R}$ is a pair $(i, p)$ that is jointly determined by the item $i\in \mathcal{C}$ being requested, and the fixed path $p$ traversed (request is forwarded from the user toward a designated source over a fixed path) to serve this request. The routing strategy of a user with respect to request $(i,p)\in\mathcal{R}$ is predetermined, e.g., the shortest path in terms of the number of hops to the nearest designated source. We assume that (i) 
the collection of requests for the same content item $i$, i.e., $\{p:\ (i,p)\in\mathcal{R} \}$,  are served separately instead of being aggregated, 
(ii) the response of request $(i,p)$ travels the same path $p$, in the reverse direction,
(iii) different frequency bands are used for the uplink and downlink
,
(iv) transmission delays are solely due to response messages carrying desired items 
assuming that request forwarding and cache downloads are instantaneous.

Request rates are known a priori, where choices of requested items are independent. The arrivals of requests are Poisson where the arrival rate of $r=(i,p)$ is $\lambda_{(i,p)}$. A path $p$ on $\mathcal{G}$ of length $|p| = K$ is a sequence $\{p_1,p_2,\hdots,p_K\}$ of nodes $p_k \in V$ such that edge $(p_k, p_{k+1}) \in E$, for $k \in\{1,\hdots, |p|-1\}$. 
Let $k_p(v)=\{k\in\{1,\hdots,|p|\}:\,p_k=v\}$ denote the position of $v$ in $p$. For each request $(i,p)$, $p_1$ is the requesting user and $p_{|p|}$ is the designated source of item $i$, and we assume that $p$ is a simple path, i.e., $p$ contains no loops. 

End-to-end delay includes several key components, such as transmission delay, propagation delay, processing delay, and queueing delay. 
In this paper, we primarily focus on lightly loaded systems, where transmission delay is the dominant component and the other delay components are negligible. 
%
%
We assume there is one queue for each link $(u,v)\in E$ that serves in a first-in-first-out (FIFO) manner all requests traversing $(u,v)$. 

To determine the transmission delay of link $(v,u)\in E$ corresponding to request $(i,p)$, we first derive the signal-to-interference-plus-noise ratio (SINR) on link $(v,u)$, which we denote by ${\SINR}_{vu}(S)$, where $S=[s_{vu}]\in {\mathbb{R}}^{|E|}$ represents the set of transmission powers at all links $(v,u) \in E$. 
To decode the requests $(i,p)$ traversing link $(u,v)$, 
we calculate the SINR on link $(v,u)$, where we treat all other transmissions from nodes $j\in V\backslash v$ 
, as well as the transmissions from $v$ to $w\neq u$ as noise. 
Therefore, the SINR on link $(v,u)$ is given as
\begin{align}
\label{SINR_general_expression}
{\SINR}_{vu}(S)&=\frac{ G_{vu}s_{vu}}{N_u+  \sum\limits_{j\in V\backslash v}G_{ju}\sum\limits_{w}s_{jw}+G_{vu}\sum\limits_{w\neq u}s_{vw}},
\end{align}
where $N_u$ is the receiver noise power at node $u$, and $s_{vu}$ is the transmit power from $v\in V$ to $u$. The total transmit power of node $v$ is $\sum\nolimits_{u:(v,u)\in E} s_{vu}$. 
The parameter $G_{vu}$ is the channel power gain that includes only path loss, where we use the standard power loss propagation model, i.e., $G_{vu}=r_{vu}^{-n}$ given distance $r_{vu}$ between $v$ and $u$, and the path loss exponent $n>2$ \cite{rappaport1996wireless}. 
The signal for request $(i,p)$ over link $(v,u)$ is decoded regarding all other signals as noise, for all $(i,p) \in {\mathcal R}$ and $(v,u) \in E$. Thus, in our model the transmission delays are coupled, in contrast to \cite{Shanmugam2013}, \cite{IoannidisYeh2017}, because the decoding model captures the interference due to simultaneous wireless transmissions. 
Because the SINR analysis in (\ref{SINR_general_expression}) is for a single frequency band, the set of active nodes with nonzero transmission powers causes interference to the unintended receiver node. Employing OFDMA-based schemes allows frequency multiplexing by moving the interfering nodes to orthogonal resources and eliminates the out-of-band interference, and improves the SINR quality. However, we leave this extension to future work.

To model the wireless transmission delay on link $(v,u)\in E$, we use the following composite relation\footnote{Practical adaptive modulation and coding schemes operate at lower $\SINR$ values 
\cite[Ch. 4.2, Ch. 9.3]{Goldsmith2005}. 
For example, for $\MQAM$ 
the gap from the Shannon $\SNR$ as function of the symbol error probability $P_e$ is $\Gamma=\frac{1}{3}(Q^{-1}(P_e))^2$. 
}:
\begin{align}
\label{wireless_delay}
f({\SINR}_{vu}(S)) = \frac{1}{\log_2(1+{\SINR}_{vu}(S))},
\end{align}
which is the delay in number of channel uses per bit corresponding to the data rate of link $(v,u)$. This model captures interference, and thus provides a more sophisticated way of modeling delay in a lightly loaded network than simple hop count. 
When the SINR is high, (\ref{wireless_delay}) yields a low transmission delay and vice versa. 
From (\ref{SINR_general_expression})-(\ref{wireless_delay}), it is clear that $f({\SINR}_{vu}(S))$ is convex and decreasing in ${\SINR}_{vu}(S)$ but non-convex in $S$.

Our goal is to jointly optimize the transmission power allocations along with the caching decisions to minimize the average transmission delay of requested items over the multi-hop network. We next formulate this problem.


\begin{table*}[t!]\small
\setlength{\extrarowheight}{2pt}
\begin{center}
\begin{tabular}{l | l }
{\bf Definition} & {\bf Symbol}\\ 
\hline
Cache capacity of $v\in V$; Catalog size; Binary caching variables & $c_v$ ; $|\mathcal{C}|$; $X=[x_{vi}]_{v\in V,\, i\in\mathcal{C}}$\\
Path of length $|p| = K$ corresponding to request $r=(i,p)$ 
& $p=\{p_1,\hdots, p_K\}$, $p_k\in V$\\
Arrival rate of request $(i,p)\in\mathcal{R}$ 
; Requests of different types (item, path)
& $\lambda_{(i,p)}>0$; $(i,p)\in\mathcal{R}$\\ 
Distance from $v$ to $u$ & $r_{vu}$\\ 
Path loss exponent & $n>2$\\
Designated sources for $i\in\mathcal{C}$ & $\mathcal{S}_i\subseteq V$\\
Total transmit power of node $v\in V$ & $\hat{s}_v$\\
Noise power at receiver $u\in V$ & $N_u$\\
SINR function on link $(v,u)\in E$; Delay function on link $(v,u)\in E$ & ${\SINR}_{vu}(S)$; $f({\SINR}_{vu}(S))\geq 0$\\
Global minimum objective and solution of the original problem (\ref{OriginalDelayOptimization}) & $D^*$; $\left(Y^*,S^*\right)$\\
Global minimum objective and solution of the RCF problem (\ref{mindelay_nonlinearproblem}) & $D^{**}$; $\left(Y^{**},S^{**}\right)$\\
Local minimum objective and solution of (\ref{mindelay_nonlinearproblem}) generated by Algorithm \ref{Projected_Subgradient_Algo} & $D^*_{sub}$; $\left(\boldsymbol{y}^*_{sub},S^*_{sub}\right)$\\
\hline
\end{tabular}
\end{center}
\caption{Notation.}
\label{table:tab1}
\end{table*}


\section{Joint Power Control and Caching Optimization for Transmission Delay Minimization}
\label{delayminimization}

In this section, we formulate the delay minimization problem that jointly considers power control and caching allocations. Due to its NP-hard nature, in Sect. \ref{caching_optimization} we first develop a RCF based on convex relaxation and its optimal solution, which yields an integral solution (via rounding) whose cost is within a constant factor from that of the optimal solution to the original problem. Next in Sect. \ref{power_optimization} we provide a sufficient condition for the convexity of RCF in the logarithm of powers which yields a biconvex objective. This sufficient condition corresponds to the high SINR regime. 
Later in Sect. \ref{joint_optimization} we jointly optimize RCF, first under the assumption of biconvexity so as to provide an alternating optimization formulation, and second under the general setting which is not jointly convex, we provide various results on the RCF objective. We demonstrate {\bf a)} strict quasi-convexity of $D(Y,S)$, {\bf b)} necessary conditions for optimality of $D(Y,S)$, {\bf c)} generalized necessary conditions for optimality of $D(Y,S)$ under strict convexity of $\mathcal{D}_S$, and {\bf d)} Pareto optimality of the solution to $D(Y,S)$. 
Finally in Sect. \ref{algorithms_optimization} we provide a subgradient projection algorithm that attains the necessary conditions, along with a linear convergence rate guarantee.

\subsection{Caching Optimization for RCF}
\label{caching_optimization}
A goal in caching systems is to minimize the expected total file downloading delay, i.e., the expected delivery time of content items averaged over the demands and the cache placement. Since end-to-end delay in our setup is mainly due to the transmission delay, by letting matrix $X = [x_{vi}]  \in \{0,1\}^{|V |\times |C|}$ denote the global caching strategy, we can express the cost function for serving a request $(i,p)$ in terms of the  transmission delay as
\begin{align}
\label{original_cost}
    D_{(i,p)}^o(X,S)=\sum\limits_{k=1}^{|p|-1}f({\SINR}_{p_{k+1}p_k}(S))\prod\limits_{l=1}^k (1-x_{p_l i})
\end{align}
where $D_{(i,p)}^o(X,S)$ includes the transmission delay of an edge $(p_{k+1}, p_k)$ in the path $p=\{p_1, \hdots p_k\}$ if none of the nodes $p_1, \hdots p_k$ caches $i$. If the request is well-routed, no edge (or cache) appears twice in (\ref{original_cost}). The last node of $p$ is the designated source, hence a request is always served. Let $D^o$ be the aggregate expected cost in terms of the average number of channel uses per bit, which equals
\begin{align}
\label{original_aggregate_cost}
D^o(X,S)=\sum\limits_{(i,p)\in\mathcal{R}}{\lambda_{(i,p)}{D_{(i,p)}^o(X,S)}}.
\end{align}
The gain of intermediate caching is equivalent to the achievable reduction in the overall transmission delay. An upper bound on the expected cost is obtained when all requests are served by the designated sources at the end of each path, i.e.,
\begin{align}
\label{Dub}
   D^{\rm{ub}}(S)=\sum\limits_{(i,p)\in\mathcal{R}}{\lambda_{(i,p)}{\sum\limits_{k=1}^{|p|-1}f({\SINR}_{p_{k+1}p_k}(S))}}. 
\end{align}

Our primary objective is to solve the problem 
\begin{align}
\label{OriginalDelayOptimization}
\min \{D^o(X,S): X\in \mathcal{D}_X,\, S\in\mathcal{D}_S\}, 
\end{align}
where $\mathcal{D}_X$ is the feasible set of 
$X \in \mathbb{R}^{|V |\times |C|}$ satisfying the capacity, integrality, and source constraints:
\begin{align}
\label{integrality_constraints}
\mathcal{D}_X=\Big\{\sum\limits_{i\in \mathcal{C}}{x_{vi}} &\leq c_v,\,\, \forall v \in V,\,\,
x_{vi} \in \{0, 1\},\,\, v \in V,\,\, i \in \mathcal{C};\nonumber\\
x_{vi} &= 1,\,\, \forall i \in \mathcal{C},\,\, v \in \mathcal{S}_i\Big\}.
\end{align}

The set of constraints $\mathcal{D}_S$ is for the power or resource budget. The feasible set of $S$ is specified by the individual power or resource budget for each node, namely  $\mathcal{D}_S$ is the feasible set of all $S=[s_{vu}]_{v\in V, u\in V\backslash v}\in\mathbb{R}^{|V|\times (|V|-1)}$ satisfying  
\begin{align}
\label{SourcePowerConstraints}
\mathcal{D}_S=\Big\{\sum_{u\in O_v} s_{vu} \leq \hat{s}_v,\,\,
s_{vu} \geq 0, \,\, \forall v \in V\Big\},
\end{align}
where $O_v = \{u\in V :(v,u)\in E\}$. 

Minimization of $D^o(X,S)$ subject to the set of integer constraints $X \in \mathcal{D}_X$ is NP-hard since it is a reduction from the 2-disjoint set cover problem \cite{Shanmugam2013}. Therefore, we aim to devise a centralized algorithm that produces an allocation within a constant approximation of the optimal, without prior knowledge of the network topology, edge weights, or the demand distribution. We next formulate a convex relaxation.

\paragraph{Convex Relaxation} 
To approximate the non-convex function $D^o(X,S)$, we construct a convex relaxation, following the approach of \cite{IoannidisYeh2016}, \cite{Shanmugam2013}. Suppose that $x_{vi}$, $v \in V$, $i \in \mathcal{C}$, are independent Bernoulli random variables. Let $\nu$ be the corresponding joint probability distribution defined over matrices in $\{0, 1\}^{|V |\times|C|}$, and denote by $\mathbb{P}_{\nu}[\cdot]$ and $\mathbb{E}_{\nu}[\cdot]$ the probability and expectation with respect to $\nu$, respectively.

Relaxing the integrality constraints of $X$ in (\ref{integrality_constraints}), let marginal probabilities 
\begin{align}
    \label{marginalprob}
    y_{vi} = \mathbb{P}[x_{vi} = 1] = \mathbb{E}_{\nu}[x_{vi}]\in[0,1], \quad v\in V,\quad i\in\mathcal{C}.
\end{align}
Denote the feasible set of $Y=[y_{vi}]_{v\in V, i\in \mathcal{C}}\in\mathbb{R}^{|V|\times |\mathcal{C}|}$ by 
\begin{align}
\label{C1relaxedfordelayproblem}
\mathcal{D}_Y=\Big\{\sum\limits_{i\in \mathcal{C}} y_{vi}&=c_v,\,\, v\in V,\,\,
y_{vi}\in[0,1], \,\, v\in V, \,\, i\in\mathcal{C};\nonumber\\ 
y_{vi}&=1, \,\, v\in \mathcal{S}_i, \,\, i\in\mathcal{C}\Big\},
\end{align}
representing the collection of (marginal) probabilities that $v\in V$ stores $i\in\mathcal{C}$ and satisfying the capacity and source constraints.

Using the definition of $Y$ in (\ref{marginalprob}), and from the fact that $x_{vi}$'s are independent and path $p$ is simple (no loop), 
we now observe that 
\begin{align}
\label{multi_linear_relaxed_problem}
D^o(Y,S)=
\mathbb{E}_{\nu}[D^o(X,S)]\ .
\end{align} 

The extension of $D^o$ to the domain $[0, 1]^{|V |\times|C|}$ is known as the multi-linear relaxation of the optimization problem \cite{Shanmugam2013},  where (\ref{OriginalDelayOptimization}) is relaxed to
\begin{align}
\label{YRelaxedDelayOptimization}
\min \{D^o(Y,S): Y\in \mathcal{D}_Y,\, S\in\mathcal{D}_S\}\ . 
\end{align}

Let $X^*$ and $Y^*$ be the optimal solutions to (\ref{OriginalDelayOptimization}) and (\ref{YRelaxedDelayOptimization}), respectively. Then, because the integrality constraints are relaxed in (\ref{marginalprob}), the cost with relaxed variables $Y^*$ satisfies for any $S\in\mathcal{D}_S$:
\begin{align}
\label{relaxed_comparison}
D^o(Y^*,S) \leq D^o(X^*,S).    
\end{align}

Note that the multi-linear relaxation $D^o(Y,S)$ in (\ref{multi_linear_relaxed_problem}) is non-convex. Therefore, we next approximate it by another cost function $D$ defined as follows:
\begin{align}
\label{DasFunctionofYS}
D(Y,S)=\sum\limits_{(i,p)\in\mathcal{R}}{\lambda_{(i,p)} D_{(i,p)}(Y,S)}, 
\end{align}
where the relaxed delay-cost for request  $(i,p)\in\mathcal{R}$ is
\begin{align}
\label{relaxed_delay_cost_per_request}
D_{(i,p)}(Y,S)={\sum\limits_{k=1}^{|p|-1}f({\SINR}_{p_{k+1}p_k}(S)) g_{p_k i}(Y) },   
\end{align}
where $f$ is given in (\ref{wireless_delay}) and $g_{p_k i}$ is given by
\begin{align}
\label{g_function}
g_{p_k i}(Y)
=1-\min\Big\{1,\sum\limits_{l=1}^k y_{p_l i}\Big\},\,\quad\forall\, y_{p_li}\in [0, 1].
\end{align}
From the Goemans-Williamson inequality \cite{GoeWil1994}, (\ref{DasFunctionofYS}) 
gives an upper bound on (\ref{multi_linear_relaxed_problem}). 
Due to the concavity of the $\min$ operator, i.e., $\mathbb{E}_{\nu}\left[g_{p_k i}(Y)\right]\geq g_{p_k i}(\mathbb{E}_{\nu}[Y])$, the function $g_{p_k i}(Y)$ is strictly quasi-convex (see Prop. \ref{quasi_convex}) in $Y$. In (\ref{g_function}), $g_{p_k i}(Y)$ is a piecewise linear function which is not smooth or strictly convex, and its partial derivatives\footnote{A function is piecewise continuously differentiable if each piece is differentiable throughout its subdomain, even if the whole function may not be differentiable at the points between the pieces \cite[Ch. 3]{beck2017first}. 
} do not exist everywhere. If the objective function or some of the constraint functions are non-differentiable, we can devise non-differentiable methods to optimize $D(Y,S)$, or subdifferential versions of KKT conditions \cite{ruszczynski2011nonlinear}, \cite[Ch. 6.3]{bertsekas1998nonlinear}. To address such scenarios we will detail an algorithm in Sect. \ref{algorithms_optimization}.

The approximated delay-cost $D(Y,S)$ is convex in the caching variables $Y$ due to the convexity of $g_{p_k i}(Y)$. Note that $D(Y,S)$ is nonconvex in the power variables $S$ because $f$ is nonconvex in $S$. We aim to solve the following reduced-complexity formulation (RCF) of the joint optimization problem:
\begin{align}
\label{mindelay_nonlinearproblem}
\min \{D(Y,S): Y\in\mathcal{D}_Y,\, S\in\mathcal{D}_S\}.
\end{align}

The objective function in  (\ref{mindelay_nonlinearproblem}) captures the wired backhaul that connects the core (or backbone) network to the MCs or the SCs at the edge of the network. The last node of a path $p$, i.e., $p_{|p|}$ can be a SC, MC or the last hop can be a connection from a MC (MC-to-backbone) or a SC to the backhaul (SC-to-backbone). The backhaul has the entire file catalog, and the connections from the MC or SC to the backhaul are wired. For given $(i,p)\in \mathcal{R}$, the transmission delay incurred by the edge $(p_{|p|-1},p_{|p|})\in p$ is given as
\begin{align}
f({\SINR}_{p_{|p|}p_{|p|-1}}(S))=
\begin{cases}
D_{MC},\,\, \,\mbox{ if}\,\, p_{|p|} \,\mbox{  MC backbone},\\
D_{SC},\,\,\,\mbox{ if}\,\, p_{|p|} \,\mbox{  SC backbone},
\end{cases}\nonumber
\end{align}
where the wired backhaul transmission delays $D_{MC}$ and $D_{SC}$ are fixed and known a priori, which we assume to be the same for all SCs and MCs based on \cite{mahloo2014cost}. In contradistinction to this, the transmission delays are coupled in the wireless part of the network due to the dependency of ${\SINR}_{vu}(S)$, $(v,u)\neq (p_{|p|},p_{|p|-1})$ in (\ref{SINR_general_expression}) on the power allocation $S$. 

The optimal value of $D(Y,S)$ in (\ref{mindelay_nonlinearproblem}), is guaranteed to be within a constant factor from the optimal values of $D^o(Y,S)$ in (\ref{YRelaxedDelayOptimization}), and of $D^o(X,S)$ in (\ref{OriginalDelayOptimization}). In particular, we have the following theorem.

\begin{theo}\label{optimalityguarantees}{\bf{Constant factor approximation for fixed $S$ \cite{Shanmugam2013,AgeSvi2004}.}}
For given $S$, let $Y^*$ and $Y^{**}$ be the optimal solutions that minimize $D^o(Y,S)$ and $D(Y,S)$ in (\ref{YRelaxedDelayOptimization}) and  (\ref{mindelay_nonlinearproblem}), respectively. Then,
\begin{align}
\label{optimality_guarantee_result}
D^o(Y^{*},S) \leq D^o(Y^{**},S)
\leq \frac{D^{\rm{ub}}(S)}{e}+\Big(1-\frac{1}{e}\Big)D^o(Y^{*},S).
\end{align}
\end{theo}

\begin{proof} 
See Appendix \ref{App:optimalityguarantees}.
\end{proof}

\paragraph{Rounding}\label{rounding}
To produce an integral solution to  (\ref{OriginalDelayOptimization}), we round the solution $Y^{ **}$ of (\ref{mindelay_nonlinearproblem}). 
For any given $S\in\mathcal{D}_S$ and given a fractional solution $Y \in \mathcal{D}_Y$, there is always a way to convert it to a $Y' \in \mathcal{D}_Y$ with at least one fewer fractional entry than $Y$, for which $D^o(Y',S) \leq D^o(Y,S )$ \cite{Blaszczyszyn2014}, \cite{IoannidisYeh2016}. 
Each rounding step reduces the number of fractional variables by at least $1$. Thus, the above algorithm concludes in at most $|V|\times |C|$ steps (assuming fixed power allocations), producing an integral solution $X' \in \mathcal{D}_X$ such that $D^o(X',S)\leq D^o(Y^{**},S)$ because each rounding step can only decrease $D^o$. Hence, from Theorem \ref{optimalityguarantees} 
and (\ref{relaxed_comparison}) we have the following corollary.

\begin{cor}\label{rounding_cor}
{\bf Rounding of caching for fixed $S$.} The integral solution $X'\in \mathcal{D}_X$ as a result of rounding satisfies for any $S\in\mathcal{D}_S$: 
\begin{align}
D^o(X^*,S) \leq  D^o(X',S)
\leq \frac{D^{\rm{ub}}(S)}{e}+\Big(1-\frac{1}{e}\Big)D^o(X^{*},S)\ . \nonumber
\end{align}
\end{cor}

Note that the rounding step produces a $\big(1-\frac{1}{e}\big)$-approximate solution, along with an offset of $\frac{D^{\rm{ub}}(S)}{e}$ to RCF. The offset in Cor. \ref{rounding_cor} is eliminated if instead of RCF in (\ref{mindelay_nonlinearproblem}) we use a maximum caching gain formulation which concerns the ultimate gain that can be obtained via caching at intermediate nodes, such as in \cite{Shanmugam2013} and \cite{IoannidisYeh2016}. In maximizing the caching gain, the objective function is given by the difference $D^{\rm{ub}}(S)-D(Y,S)$, where $D^{\rm{ub}}(S)$ is given by (\ref{Dub}). In this case, the relationship $D^o(X^*,S) \leq D^o(Y^{**},S) \leq \frac{D^{\rm{ub}}(S)}{e}+\big(1-\frac{1}{e}\big)D^o(X^{*},S)$ is equivalent to $D^{\rm{ub}}(S)-D^o(X^*,S) \geq D^{\rm{ub}}(S)-D^o(Y^{**},S) \geq D^{\rm{ub}}(S)-\frac{D^{\rm{ub}}(S)}{e}-\big(1-\frac{1}{e}\big)D^o(X^{*},S)=\big(1-\frac{1}{e}\big)(D^{\rm{ub}}(S)-D^o(X^{*},S))$, giving a $\big(1-\frac{1}{e}\big)$-approximate solution for the maximum caching gain formulation without an offset. However, in this formulation the difference $D^{\rm{ub}}(S)-D(Y,S)$ increases in $S$, requiring high powers. Hence, despite its offset, RCF formulation in (\ref{mindelay_nonlinearproblem}) is preferable as it can jointly optimize power.

\paragraph{$D^o$ and $D$ are not jointly convex in $Y$ and $S$} 
The transmission delays are coupled due to the interference from simultaneous transmissions. From (\ref{wireless_delay}), $f$ is not convex in $S$. Furthermore, (\ref{multi_linear_relaxed_problem}) is not convex in $Y$ for given $S$ and not convex in $S$ for given $Y$, hence not jointly convex in $(Y,S)$. 
Note that $D(Y,S)$ is jointly convex at low interference or low power because the logarithm function in (\ref{wireless_delay}) changes linearly (and its reciprocal is convex) in power when $\SINR$ is low in all paths, which is true in the power-limited regime. 

The joint convexity of $D$ requires the Hessian matrix $H$ of $D(Y,S)$ with respect to $(Y,S)$ to be positive semi-definite (PSD). 
Since (\ref{g_function}) is not differentiable, the Hessian matrix for $D(Y,S)$ with respect to $Y$, i.e., $\nabla_Y^2 D$, is not defined. However, from \cite[Theorem 2.1]{scheimberg1992descent}, the second order derivatives for maximum functions are defined in each interval and the subhessians of (\ref{DasFunctionofYS}) or (\ref{g_function}) with respect to $Y$, i.e., $\{d^2_{Y} D\}$, exist and we can define a subhessian matrix $d^2_{Y} D$. However, since (\ref{g_function}) is piecewise linear, $d^2_{Y} D$ is a zero matrix. Combining this with the Schur's complement condition for $H$ to be PSD in \cite{Boyd2009}, $D(Y,S)$ is jointly convex only if the off-diagonals of $H$ are singular. However, in our setting, the partial derivatives $\nabla_S D$ with respect to $S$ are nonzero, and the subhessian matrix formed by their subgradients with respect to $Y$ is non-singular. Therefore, $D(Y,S)$ is not jointly convex.
%
While $D(Y,S)$ is not jointly convex in $(Y,S)$ in general, it can be biconvex in the logarithms of the power variables under a condition, which we provide next in Sect. \ref{power_optimization} in Prop. \ref{sufficient_condition_for_convexity}.

\subsection{Power Optimization for RCF}
\label{power_optimization}

We next provide a sufficient condition for $f({\SINR}_{p_{k+1}p_k})$ to be convex in $\log$ power variables $P\eqdef (\log(s_{vu}))_{(v,u)\in E}$ in which $P_{vu}=\log(s_{vu})$ denotes power measured on link $(v,u)$ corresponding to request $(i,p)$ in dB. 

\begin{prop}\label{sufficient_condition_for_convexity}
{\bf Convexity in $\log$ power variables.} A sufficient condition for the composite function $f({\SINR}_{p_{k+1}p_k})$ to be convex in $P\eqdef (\log(s_{vu}))_{(v,u)\in E}$ is given as follows.
\begin{align}
\label{sufficient_condition_convex_log_powers}
    \frac{2f^{'}(x)^2}{f(x)}\cdot x-f^{'}(x)\leq f^{''}(x)\cdot x, \quad \forall x\geq 0.
\end{align}
\end{prop}

\begin{proof}
The result follows from extending the approach in \cite{huang2006distributed}. For details, see Appendix \ref{App:sufficient_condition_for_convexity}. 
\end{proof}

The sufficient condition (\ref{sufficient_condition_convex_log_powers}) of Prop. \ref{sufficient_condition_for_convexity} holds in the high $\SINR$ regime where $\log(1+\SINR)\approx \log(\SINR)$. However, (\ref{sufficient_condition_convex_log_powers}) no longer holds when $\SINR\gg 1$ does not hold. 
Given the sufficient condition in (\ref{sufficient_condition_convex_log_powers}), it is clear that the program (\ref{mindelay_nonlinearproblem}) is convex in terms of power measured in dB. Hence, we define the log-power variables $P$, belonging to the feasible set
\begin{align}
\mathcal{D}_P=\{P_{vu}\in\mathbb{R}
:\, \sum\limits_{u\in O_v} e^{P_{vu}} \leq \hat{s}_v,\,\,
\forall v\in V, \,\, \forall (u,v)\in E\},\nonumber
\end{align}
where $O_v = \{u\in V :(v,u)\in E\}$.

The condition of Prop. \ref{sufficient_condition_for_convexity} ensures that $D(Y,P)$ is biconvex, i.e., $D(Y,P)$ is convex in $Y$ for given $P$ and convex in $P$ for given $Y$
\cite{floudas2013deterministic}. 
This paves the way for employing methods to solve RCF in (\ref{mindelay_nonlinearproblem}).  We next outline one such method.

\subsection{Joint Optimization of RCF}
\label{joint_optimization}
In this section, we present two techniques to optimize RCF: {\bf 1)} Biconvex optimization of $D(Y,S)$ under the condition of Prop. \ref{sufficient_condition_for_convexity} on convexity in $\log$ powers, and {\bf 2)} General joint optimization where $D(Y,S)$ is not jointly convex in $Y$ and $S$. For the former, we exploit alternating optimization methods. For the latter, we prove various results on $D(Y,S)$: {\bf a)} strict quasi-convexity, {\bf b)} necessary conditions for optimality, {\bf c)} generalized necessary conditions under strict convexity of $\mathcal{D}_S$, and {\bf d)} Pareto optimality of $D(Y,S)$. 

\subsubsection{Alternating Optimization}

We next present a biconvex optimization technique for RCF. To that end, we exploit alternating optimization methods. 
There exist techniques to find the local optimum of biconvex minimization problems, such as block-relaxation methods \cite{gorski2007biconvex}. Furthermore, the global optimum of biconvex problems can be determined for certain classes of constraints \cite{floudas1990global}.

Provided that the convexity condition in Prop. \ref{sufficient_condition_for_convexity} holds, $D(Y,S)$ is biconvex and hence we can focus on the alternating optimization of RCF. This corresponds to alternatively updating the power variables $S$ given the caching variables $Y$, and then updating $Y$ given $S$. This iterative optimization approach can find a local optimum to the average delay minimization problem. To obtain an integral solution, the algorithm needs a rounding step before it terminates. This technique for RCF is summarized in Algorithm \ref{alternating}. An algorithm called Global OPtimization (GOP) algorithm was developed in \cite{floudas1990global} to exploit the convex substructure of constrained biconvex minimization problems by a primal-relaxed dual approach. The objective function and the constraints in RCF satisfy the necessary convexity conditions \cite[Ch. 3.1, Conditions (A)]{floudas2013deterministic} for the GOP algorithm. However, \cite[Ch. 3.1]{floudas2013deterministic}, \cite[Theorem 1, Condition (d)]{floudas1990global} require the multipliers for the primal problem to be uniformly bounded, which may not be true for RCF. Hence, employing the GOP algorithm does not guarantee termination in a finite number of steps for any $\epsilon>0$  \cite[Theorem 3.6.1]{floudas2013deterministic}, or at the global optimum of (\ref{DasFunctionofYS}) \cite[Theorem 3.6.2]{floudas2013deterministic}.

We note that the proposed alternating approach requires the condition in Prop. \ref{sufficient_condition_for_convexity}, while no optimality guarantee is  established. However, $D(Y,S)$ is not necessarily biconvex because it is nonconvex in $S$ when interference is non-negligible, i.e., at low $\SINR$. Deriving the necessary conditions for optimality will reveal the true potential of the algorithm and elucidate the effect of network's operating regime, e.g., in the high or low $\SINR$.

\algdef{SE}[DOWHILE]{Do}{doWhile}{\algorithmicdo}[1]{\algorithmicwhile\ #1}%
\begin{algorithm}[t!]\small
\caption{\fontfamily{ptm}\selectfont Alternating optimization for biconvex $D(Y,S)$.}\label{alternating}
\begin{algorithmic}[1]
\BState {\bf{Begin: $S^{0} \in \mathcal{D}_S; Y^{0} \in \mathcal{D}_Y$ ; }}\\
Let $t = 0$ ;
\Do 
\State{$Y^{t+1}=\arg\min \limits_{Y}  D(Y,S^{t})$ (convex with start point $Y^{t}$) 
}
\State{$S^{t+1}=\arg\min \limits_{S}  D(Y^{t+1},S)$ (convex with start point $S^{t}$)}
\State{Let $t = t+1$;}
\doWhile{$D(Y^{t},S^{t}) - D(Y^{t-1},S^{t-1}) > \epsilon$}
\State{Let $(Y^{**},S^{**}) = (Y^{t},S^{t})$;} 
\State{Implement {\em b) Rounding}.}
\end{algorithmic}
\end{algorithm}

\subsubsection{General Joint Optimization}

We next extend the approach of \cite{IoannidisYeh2016} to develop centralized algorithms for the joint power-caching optimization of RCF which is not biconvex, i.e., the sufficient condition 
in $\log$ powers imposed by Prop. \ref{sufficient_condition_for_convexity} does not hold.

We first present a general result on the relaxed cost function $D(Y,S)$ without putting any assumptions on the $\log$ powers or the caching variables.

\paragraph{Strict quasi-convexity of $D(Y,S)$}

\begin{prop}\label{quasi_convex}
The 
relaxed delay-cost function $D(Y,S)$ of RCF in (\ref{mindelay_nonlinearproblem}) is strictly quasi-convex. 
\end{prop}

\begin{proof}
See Appendix \ref{App:quasi_convex}.
\end{proof}

Note that the partial derivatives of the relaxed delay-cost function $D_{(i,p)}(Y,S)$,  $(i,p)\in\mathcal{R}: (u,v)\in p$ with respect to $s_{vu}$ and the subgradients of $D_{(i,p)}(Y,S)$ with respect to $y_{vi}$ satisfy
\begin{align}
\label{D_vs_S} 
\frac{\partial D_{(i,p)}}{\partial s_{ju}} &\overset{(a)}{\geq} 0,\,\,\,\,
\frac{\partial^2 D_{(i,p)}}{\partial s^2_{ju}} \overset{(a)}{\leq} 0, \,\, j\in I_u,\\
\frac{\partial D_{(i,p)}}{\partial s_{vu}} &\overset{(b)}{\leq} 0, \quad \frac{\partial^2 D_{(i,p)}}{\partial s^2_{vu}} \overset{(b)}{\geq} 0,\nonumber\\
\label{D_vs_Y}
d_{y_{mi}} D_{(i,p)} &\overset{(c)}{\leq} 0, \quad 
d^2_{y_{mi}} D_{(i,p)} \overset{(d)}{\geq} 0,\,\, m\in p,
\end{align}
where $(a)$ follows from that $f({\SINR}_{vu}(S))$ is a decreasing function of ${\SINR}_{vu}(S)$ which is decreasing in $s_{ju}$ for $j\in I_u$, and similarly $(b)$ from that ${\SINR}_{vu}(S)$ is linearly proportional to $s_{vu}$ and $f({\SINR}_{vu}(S))$ is inversely proportional to $\log(1+{\SINR}_{vu}(S))$ and convex in ${\SINR}_{vu}(S)$. Note that $(c)$ follows from (\ref{g_function}), and $(d)$ from the convexity of $D(Y,S)$ in $Y$.


\paragraph{Necessary conditions for optimality of $D(Y,S)$}\label{KKT}

We investigate the necessary conditions, i.e., the Karush–Kuhn–Tucker (KKT) conditions, for a solution of $D(Y,S)$ to be optimal. Assume that $D(Y,S)$ and the constraints are continuously differentiable at $(Y^{**},\,S^{**})$. If  $(Y^{**},\, S^{**})$ gives a local optimum and the optimization problem satisfies some regularity conditions \cite{bertsekas1998nonlinear}, then there exist constants $[\mu_{v,i}]_{v\in V, i\in \mathcal{C}}$, $[\nu_{v,i}]_{v\in V, i\in \mathcal{C}}$, $[\eta_{v}]_{v\in V}$, $[\beta_{v}]_{v\in V}$, $[\gamma_{e,r}]_{e=(u,v)\in E, r=(i,p)\in \mathcal{R}}$ called KKT multipliers, such that the followings hold.

\begin{theo}\label{necessary_optimal}
{\bf Necessary conditions for optimality of $D(Y,S)$.} For a feasible set of power and cache allocations $[s_{vu}]_{(u,v)\in E}$, and $[y_{vi}]_{v\in V, i\in \mathcal{C}}$ to be the solution of RCF in (\ref{mindelay_nonlinearproblem}), the following conditions are necessary.

There exist subgradients for the caching variables that satisfy 
\begin{align}
\label{caching_subgradients}
d_{y_{vi}} D&=\alpha_{vi},\,\, {\rm{if}}\,\, y_{vi}\in (0,1),\nonumber\\ 
d_{y_{vi}} D&\geq \alpha_{vi},\,\, {\rm{if}}\,\, y_{vi}=0,\nonumber\\
d_{y_{vi}} D&<\alpha_{vi},\,\, {\rm{if}}\,\, y_{vi}=1,
\end{align}
where $\alpha_{vi}$, $v\in V$, $i\in \mathcal{C}$ is some constant.

The gradients for the power variables should satisfy 
\begin{align}
\label{s_slackness}
\frac{\partial D}{\partial s_{vu}}&\geq -\beta_v+\gamma_{e,r},\,\, {\rm{if}}\,\, s_{vu}=0,\nonumber\\ 
\frac{\partial D}{\partial s_{vu}}&= -\beta_v,\,\, {\rm{if}}\,\,  s_{vu}>0\,\,{\rm{and}}\,\, \sum\limits_{u\in O_v}  s_{vu} = \hat{s}_v,\nonumber\\ 
\frac{\partial D}{\partial s_{vu}}&= 0,\,\, {\rm{if}}\,\, s_{vu}>0\,\,{\rm{and}}\,\, \sum\limits_{u\in O_v} s_{vu} < \hat{s}_v,
\end{align}
for nonnegative constants $\beta_v$, $v\in V$, $\gamma_{e,r}$, $e=(u,v), r\in \mathcal{R}$.
\end{theo}

\begin{proof}
See Appendix \ref{App:necessary}.
\end{proof}

Note that when  $D_{(i,p)}(Y,S)$ in (\ref{relaxed_delay_cost_per_request}) is jointly convex in $(Y,S)$ (which is not true in general and requires a more restrictive condition than Prop. \ref{sufficient_condition_for_convexity} on the power control variables), the conditions in Theorem \ref{necessary_optimal} are also sufficient for optimality of $D(Y,S)$ \cite[Theorem 1]{XiYeh2008}. 

The following characterizes the optimality conditions for the relaxed delay-cost function $D(Y,S)$ with a general convex power allocation region $\mathcal{D}_S$ which is true from linearity of (\ref{SourcePowerConstraints})
, and a general convex cache allocation region $\mathcal{D}_Y$. 

\paragraph{\emph Generalized KKT conditions that requires strictly convex $\mathcal{D}_S$ for unique optimal solution}

\begin{prop}\label{necessary}
Assume that the cost functions $D_{(i,p)}(Y,S)$ satisfy (\ref{D_vs_S}) and (\ref{D_vs_Y}), and $\mathcal{D}_S$ is convex. Then, for a feasible set of cache and power allocations $(y_{vi})_{v\in V,\, i\in\mathcal{C}}$ and $(s_{vu})_{(v,u)\in E}$ to be a solution of (\ref{mindelay_nonlinearproblem}), the following conditions are necessary: 

For all $v\in V$, $i\in\mathcal{C}$, there exists a constant $\alpha_{vi}$ for which 
\begin{align}
\label{necessary_y_constraint}
d_{y_{vi}} D&=\alpha_{vi},\,\, {\rm{if}}\,\, y_{vi}\in (0,1),\nonumber\\
d_{y_{vi}} D&\geq \alpha_{vi},\,\, {\rm{if}}\,\, y_{vi}=0,\nonumber\\
d_{y_{vi}} D&<\alpha_{vi},\,\, {\rm{if}}\,\, y_{vi}=1.
\end{align}

For all feasible $(\Delta s_{vu})_{(v,u)\in E}$ at $(s_{vu})_{(v,u)\in E}$
\begin{align}
\label{necessary_s_constraint}
    \sum\limits_{(i,p)\in\mathcal{R}} \frac{\partial D_{(i,p)}}{\partial s_{vu}}(Y,S)\cdot \Delta s_{vu} &\geq 0,\\
\label{necessary_s_interference_constraint}    
    \sum\limits_{(i,p)\in\mathcal{R}} \frac{\partial D_{(i,p)}}{\partial s_{ju}}(Y,S^{**})\cdot \Delta s_{ju}&\geq 0,\,\, j\in I_u,
\end{align}
where $\Delta s_{vu}$ at $s_{vu}$ is an incremental direction which is feasible if there exists $\bar{\delta}>0$ such that $(s_{vu}+\delta\cdot \Delta s_{vu})\in\mathcal{D}_S$ for any $\delta\in (0,\bar{\delta})$.

\end{prop}

\begin{proof}
The necessary 
conditions follow from the arguments in Theorem \ref{necessary_optimal}. However, we still need to detail why (\ref{necessary_s_constraint}) is true. 
By the convexity of cost functions, the cost difference of two configurations $(Y,S^a)$ and $(Y,S^{**})$ for any feasible $S^a$ is
\begin{align}
    &\sum\limits_{(i,p)\in\mathcal{R}}  D_{(i,p)}(Y,S^a)-\sum\limits_{(i,p)\in\mathcal{R}}  D_{(i,p)}(Y,S^{**}) \nonumber\\
    &\geq  \sum\limits_{(i,p)\in\mathcal{R}} \frac{\partial D_{(i,p)}}{\partial s_{vu}}(Y,S^{**}) (s^a_{vu}-s^{**}_{vu}) 
    \geq 0,\nonumber
\end{align}
where the last inequality follows from the complementary slackness condition given in (\ref{s_slackness}), i.e., $\frac{\partial D_{(i,p)}}{\partial s_{vu}}(Y,S^{**})=0$ since $s^{**}_{vu}>0$, and $\sum\limits_{u\in O_v} \sum\limits_{(i,p):(u,v)\in p} s^{**}_{vu}(i,p) < \hat{s}_v$. Furthermore,
\begin{align}
    &\sum\limits_{(i,p)\in\mathcal{R}}  D_{(i,p)}(Y,S^a)-\sum\limits_{(i,p)\in\mathcal{R}}  D_{(i,p)}(Y,S^{**}) \nonumber\\
    &\geq  \sum\limits_{(i,p)\in\mathcal{R}} \frac{\partial D_{(i,p)}}{\partial s_{ju}}(Y,S^{**}) (s^a_{ju}-s^{**}_{ju})\geq 0,\,\, j\in I_u,\nonumber    
\end{align}
where the last inequality also follows from the complementary slackness condition in (\ref{s_slackness}).
\end{proof}

If $D_{(i,p)}(Y,S)$ is jointly convex in $(Y,S)$, the above conditions are also sufficient when (\ref{necessary_y_constraint}) holds for all $v\in V$. Furthermore, the optimal $S^{**}$ is unique if $\mathcal{D}_S$ is strictly convex. Moreover, if $D_{(i,p)}(Y,S)$ is strictly convex in $Y$, then the optimal cache allocations for the relaxed cost function $Y^{**}$ are unique as well. We do not prove this statement. However, it can be proven using arguments similar to the those in \cite[Theorem 3]{XiYeh2008}.

\paragraph{Pareto optimality of $D(Y,S)$}

When $f({\SINR}_{vu}(S))$ is chosen to be (\ref{wireless_delay}), we infer that the sufficiency part of Theorem \ref{necessary_optimal} does not hold since $D_{(i,p)}(Y,S)$ is in general not jointly convex in $(Y,S)$. Hence, we further need to establish the conditions for a Pareto optimal operating point for strictly quasi-convex cost functions (as shown in Prop. \ref{quasi_convex}). We next show that for a solution $(Y^{**},\, S^{**})$ that both satisfies (\ref{necessary_y_constraint}) and (\ref{necessary_s_constraint}), we have the following Pareto optimal property.


\begin{theo}
\label{pareto}{\bf Pareto optimality of $D(Y,S)$.} 
From Prop. \ref{quasi_convex} on the strict quasi-convexity we have $f({\SINR}_{vu}(S))$ in (\ref{wireless_delay}), $g_{p_k i}(Y)$ in (\ref{g_function}), 
and the relaxed delay-cost function for RCF in (\ref{mindelay_nonlinearproblem}) are strictly quasi-convex. 
If a pair of feasible cache and power allocations $((y_{vi}^{**}),\, (s^{**}_{vu}))$ 
satisfies conditions (\ref{necessary_y_constraint})-(\ref{necessary_s_constraint}) \cite[Thm. 3]{XiYeh2008} simultaneously, then the vector of transmission delays $(D_{(i,p)}(Y^{**},S^{**}))_{(i,p)\in\mathcal{R}}$ is Pareto optimal, i.e., there does not exist another pair of feasible allocations $((y_{vi}^{\#}),\, (s^{\#}_{vu}))$ such that $D_{(i,p)}(Y^{\#},S^{\#})\leq D_{(i,p)}(Y^{**},S^{**})$, $\forall (i,p)\in\mathcal{R}$, with at least one inequality being strict. 
\end{theo}

Given the relaxed delay-cost function $D(Y,S)$ of the form (\ref{DasFunctionofYS}), Theorem \ref{pareto} implies that at the Pareto optimal point, the cost of a request $(i,p)\in\mathcal{R}$ cannot be strictly reduced without increasing the cost of another request $(i',p')\in\mathcal{R}$.

\begin{proof}
See Appendix \ref{App:pareto}.
\end{proof}





We next focus on devising algorithms to attain the KKT points and demonstrate their convergence.

\subsection{Algorithms to Find the Optimal Solution}
\label{algorithms_optimization}
For general SINR scenario, the sufficient condition for the convexity given in (\ref{sufficient_condition_convex_log_powers}) is not necessarily satisfied, and hence the objective function is not necessarily convex. 
Due to the non-differentiability and non-convexity of $D(Y,S)$ in general SINR condition, we adopt a subgradient projection method solving for the local minima. This is the Pareto optimal solution for $D(Y,S)$ provided that the conditions in Prop. \ref{necessary} hold. In that case, from Theorem \ref{optimalityguarantees}, at each rounding step, the subgradients will guarantee a constant factor approximation for any given $S\in\mathcal{D}_S$.

\paragraph{Algorithm overview} Let $\boldsymbol{y}$ to denote the vectorized caching variable $Y$, namely $\boldsymbol{y} \in [0,1]^{|V||\mathcal{C}| \times 1}$ with $y_{vi} = \boldsymbol{y}_{(i-1)|V|+v} , \forall v \in V, i \in \mathcal{C}$. 

For the $t$-th iteration, the subgradient projection method can be summarized by the following:
\begin{equation}
\begin{aligned}
    & S^{t+1} = S^t + \xi_S^t(\bar{S}^t - S^t),\quad
     \bar{S}^t = [S^t - w_S^t d_S^t]^+_{\mathcal{D}_S}, \\
    & \boldsymbol{y}^{t+1} = \boldsymbol{y}^t + \xi^t_{\boldsymbol{y}}(\boldsymbol{\bar{y}}^t - \boldsymbol{y}^t),
   \quad \boldsymbol{\bar{y}}^t = [\boldsymbol{y}^t - w_Y^t d^t_{\boldsymbol{y}}]^+_{\mathcal{D}_{\boldsymbol{y}}},
\end{aligned}
\label{subgradient_update}
\end{equation}
where $\xi^t_S$, $\xi^t_{\boldsymbol{y}}\in (0,1]$ are step sizes respectively corresponding to $S$ and $\boldsymbol{y}$ , $w^t_S$ and $w^t_Y$ are positive scalars, $[x]^+_A$ denotes projection of vector $x$ on a convex constraint set $A$, and  
\begin{align}
    d_S^t = \nabla_S D(Y^t,S^t), \quad d^t_{\boldsymbol{y}} \in \partial_{\boldsymbol{y}}D(Y^t,S^t)\ ,  \label{subgradient_def}
\end{align}
where $d_S^t$ and $d^t_{\boldsymbol{y}}$ are the subgradients at iteration $t$  with respect to $S$ and $\boldsymbol{y}$, respectively. $\partial_{\boldsymbol{y}}D(Y^t,S^t)$ is the subdifferential with respect to $\boldsymbol{y}$.

\paragraph{Subgradient} 
Note that since $D(Y,S)$ is continuously differentiable in $S$ over set $\mathcal{D}_S$, the subdifferential of $D(Y,S)$ with respect to $S$ will only contain the gradient. Meanwhile, $\partial_{\boldsymbol{y}}D(Y^t,S^t)$ could be explicitly calculated by 
evaluating $\partial_{y_{vi}} g_{p_k i}$'s inside the term (\ref{relaxed_delay_cost_per_request}) and using (\ref{DasFunctionofYS}), 
where\begin{align*}
    \partial_{y_{vi}} g_{p_k i} =  \begin{cases}
    \{1\}, \quad &\text{if }\sum_{l = 1}^k y_{p_l i} < 1\ ,
    \\ \{0\},  \quad &\text{if } \sum_{l = 1}^k y_{p_l i} > 1\ ,
    \\ [0,1], \quad &\text{if }\sum_{l = 1}^k y_{p_l i} = 1\ .
    \end{cases}
\end{align*}

\paragraph{Step size}
The gradient/subgradient magnitudes might be significantly different for $Y$ and $S$, and therefore we calculate their step sizes separately. Note that 
$D(Y,S)$ is not \emph{Lipschitz continuous} in $S$ \cite[Sect. 1.2.2]{bertsekas1998nonlinear} and a constant step size will not provide a convergence guarantee.

We instead use a modified Polyak’s step size 
\cite{polyak1987introduction}. Let $D^t = D(\boldsymbol{y}^t,S^t)$, then
\begin{align}
    \xi^t_{\boldsymbol{y}} = \frac{D^t - \hat{D}^t}{\norm{d_{\boldsymbol{y}}^t}^2},\quad
    \xi_S^t = \frac{D^t - \hat{D}^t}{\norm{d_S^t}^2} \label{stepsize}
\end{align}
where 
$\hat{D}^t = \min_{j = 0,\cdots,t} D(\boldsymbol{y}^t,S^t) - \delta_t$ is an estimate of the local minima, $\{\delta_t\}_{t\geq 0}$ is a sequence of positive scalars 
satisfying $\lim_{t \to \infty}\delta_t = 0$ and $\lim_{t\to \infty}\sum_{m = 0}^t \delta_m = \infty$. Then the subgradient algorithm is guaranteed to converge to a local minima $D^*$.

\paragraph{Convergence rate}

Using the modified Polyak’s step size \cite{polyak1987introduction} in (\ref{stepsize}), the subgradient projection algorithm is guaranteed to converge to a local minima $D^*_{sub}$, which we provide next.

\begin{lem}
\label{LemmaConvergence}
Let $(\boldsymbol{y}^t,S^t)$ be generated by the subgradient projection algorithm with modified Polyak's step size (\ref{stepsize}). Then, the algorithm converges to a local minima $D^*_{sub}$, i.e.,
 \begin{align}
     \liminf_{t \to \infty} D^t = D^*_{sub}\ .
 \end{align}
\end{lem}

\begin{proof}
See Appendix \ref{App:LemmaConvergence}.
\end{proof}


\begin{algorithm}[b!]\small
\caption{\fontfamily{ptm}\selectfont Projected Subgradient Method}
\label{Projected_Subgradient_Algo}
\begin{algorithmic}[1]
\State{Choose $S^{0}$, $\boldsymbol{y}^{0}$, small scalar $\epsilon > 0$ and let $t=0$ }
\Do
    \State{Compute subgradient $d^t_S, d^t_{\boldsymbol{y}}$ by (\ref{subgradient_def})}
    \State{Determine step sizes $\xi^t_{\boldsymbol{y}}$, $\xi_S^t $ according to (\ref{stepsize})}
    \State{Compute projected variables $\boldsymbol{\bar{y}}^t$ and $\bar{S}^t $ by (\ref{subgradient_update})}
    \State{Update $S^{t+1}$ and $\boldsymbol{y}^{t+1}$ by (\ref{subgradient_update})}
    \State{Let $t = t+1$}
\doWhile{$D^t - D^{t-1} > \epsilon$}
\State{Let $(\boldsymbol{y}^{*}_{sub},S^{*}_{sub}) = (\boldsymbol{y}^{t},S^{t})$}
\State{Implement {\em b) Rounding}.}
\end{algorithmic}
\end{algorithm}

The subgradient projection algorithm converges 
linearly. To see this, define $(\boldsymbol{y}^*_{sub},S^*_{sub})$ to be the set of $\boldsymbol{y}$ and $S$ that attains the local minima $D^*_{sub}$.
Given the objective $D(\boldsymbol{y},S)$ and its subgradients are bounded near $(\boldsymbol{y}^*_{sub},S^*_{sub})$, we say $D(\boldsymbol{y},S)$ has a \emph{sharp set of minima} near and inside $(\boldsymbol{y}^*_{sub},S^*_{sub})$ (see Proof of Lemma \ref{LemmaConvergenceRate}), namely there exists $\mu > 0$ such that for any $S\in \mathcal{D}_S$ and $\boldsymbol{y} \in \mathcal{D}_{Y}$,
\begin{equation}
\label{SharpSetofMinima}
    D(\boldsymbol{y},S) - D^*_{sub} \geq \mu L(\boldsymbol{y},S)\ ,
\end{equation}
where  
\begin{align}
\label{L_function}
L(\boldsymbol{y},S) = \min_{\substack{y \in \boldsymbol{y}^{*}_{sub} \\ s \in S^*_{sub}}}{\sqrt{\norm{\boldsymbol{y} - y}^2 + \norm{S - s}^2}}
\end{align}
denotes the distance from $(\boldsymbol{y},S)$ to the set $(\boldsymbol{y}^*_{sub},S^*_{sub})$. 

The existence of \emph{sharp set of minima} further leads to a bound of improvement by each step, i.e.,
\begin{align*}
    \left(L^{t+1}\right)^2 &\leq \left(L^{t}\right)^2
    - \frac{D^t - D^*_{sub}}{U^2}.
\end{align*}
where $L^t = L(\boldsymbol{y}^{t},S^{t})$. Lemma \ref{LemmaConvergenceRate} below then follows by aggregating the bounds for iteration $0$ through $t-1$.

\begin{lem}
\label{LemmaConvergenceRate}
Let $(\boldsymbol{y}^{t},S^t)$ be generated by the subgradient projection algorithm with modified Polyak's step size in (\ref{stepsize}). Then, $L$ linearly converges according to 
\begin{equation}
    L(\boldsymbol{y}^{t},S^t) \leq \left(1 - \frac{\mu^2}{U^2}\right)^\frac{t}{2} L(\boldsymbol{y}^0,S^0) \label{convergelinearly}
\end{equation}
where $\mu$ is the finite positive scalar in (\ref{SharpSetofMinima}), and $U$ is a finite positive scalar with $\norm{d^t_{\boldsymbol{y}}}^2 + \norm{d^t_S}^2 \leq U$ for any $t$.
\end{lem}

\begin{proof}
See Appendix \ref{App:LemmaConvergenceRate}.
\end{proof}

We summarize the subgradient projection method that achieves the local minima in Algorithm \ref{Projected_Subgradient_Algo}.

\begin{figure}[t!]
    \centering
    \includegraphics[width=0.9\linewidth]{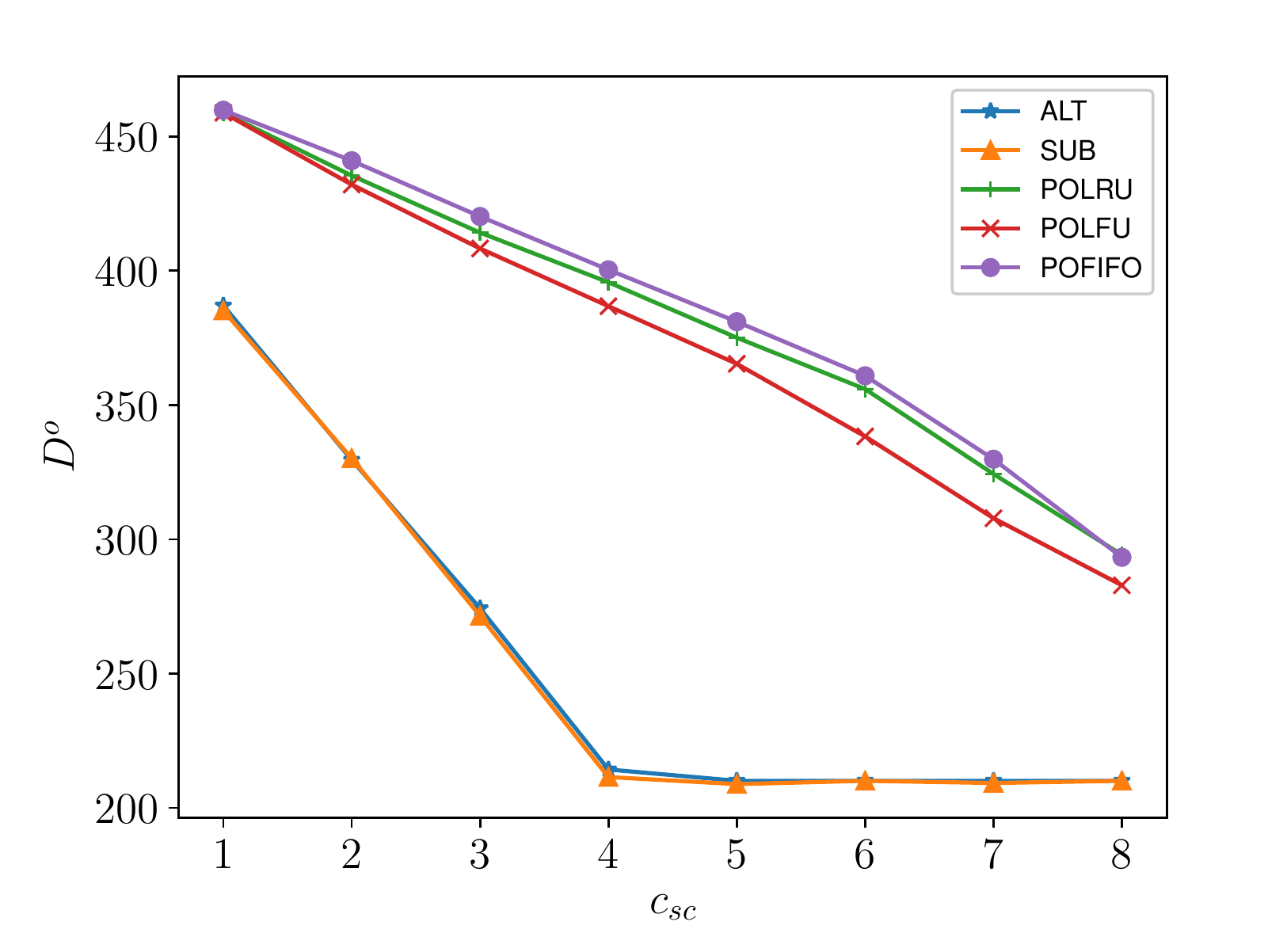}
    \vspace{-0.25cm}
    \caption{$D^o$ versus increasing cache capacity. All SC cache capacities are equal to $c_{sc}$, MC cache capacity is $c_{mc} = \min(2c_{sc},8)$, $\gamma = 0.25$ and $\hat{s}_v = 100$.}
    \label{fig:s1cache}
\end{figure}

\section{Numerical Results}
\label{numericalresults}
In this section, we present numerical results obtained from several simulation scenarios. We simulate a network in accordance with the model in Sect. \ref{systemmodel} and compare the performance of Algorithms \ref{Projected_Subgradient_Algo} (SUB) and \ref{alternating} (ALT) 
to the LRU, LFU and FIFO cache replacement policies. We pair these policies with power optimization to have a fair comparison. To make this distinction clear, we name these power optimal (PO) policies POLRU, POLFU and POFIFO when reporting results. 

{\em Simulation setup.} We simulate a network with 30 users, 4 SCs and a single MC. Users are distributed uniformly, while SCs are distributed using Lloyd's algorithm \cite{lloyd1982least}, inside the coverage area of the MC. The users do not cache items, and each one requests a single item at a given time, from a catalog of 10 items, based on a Zipf distribution with parameter $\gamma$ which can be interpreted as the popularity distribution of content items. The backhaul is the source for all items while the MC and SCs are not designated sources for any item. When a request for an item arrives at the MC or an SC, if the item is not already cached there, it is retrieved from an uplink node that caches the item or from the backhaul and then cached. We calculate gains using pathloss exponent $n = 3.7$ and we set noise power to $N_u = 1$ for all $u \in V$. For Algorithms 
\ref{Projected_Subgradient_Algo} and \ref{alternating}, 
we set the initial points $S^0$ and $Y^0$ so that $s^0_{vu} = \hat{s}_v / |O_v|$ and $y^0_{vi} = 0$ for all $v, u \in V$ and $i \in \mathcal{C}$. While the algorithms we propose can optimize a snapshot of the network, LRU, LFU and FIFO policies assume a cache history. Therefore, we simulate these policies in a time-slotted fashion and compare their average results to our algorithms. We now discuss our observations from three distinct simulation settings. We include any other necessary parameters and details in these discussions.

\begin{figure}[t!]
    \centering
    \includegraphics[width=0.9\linewidth]{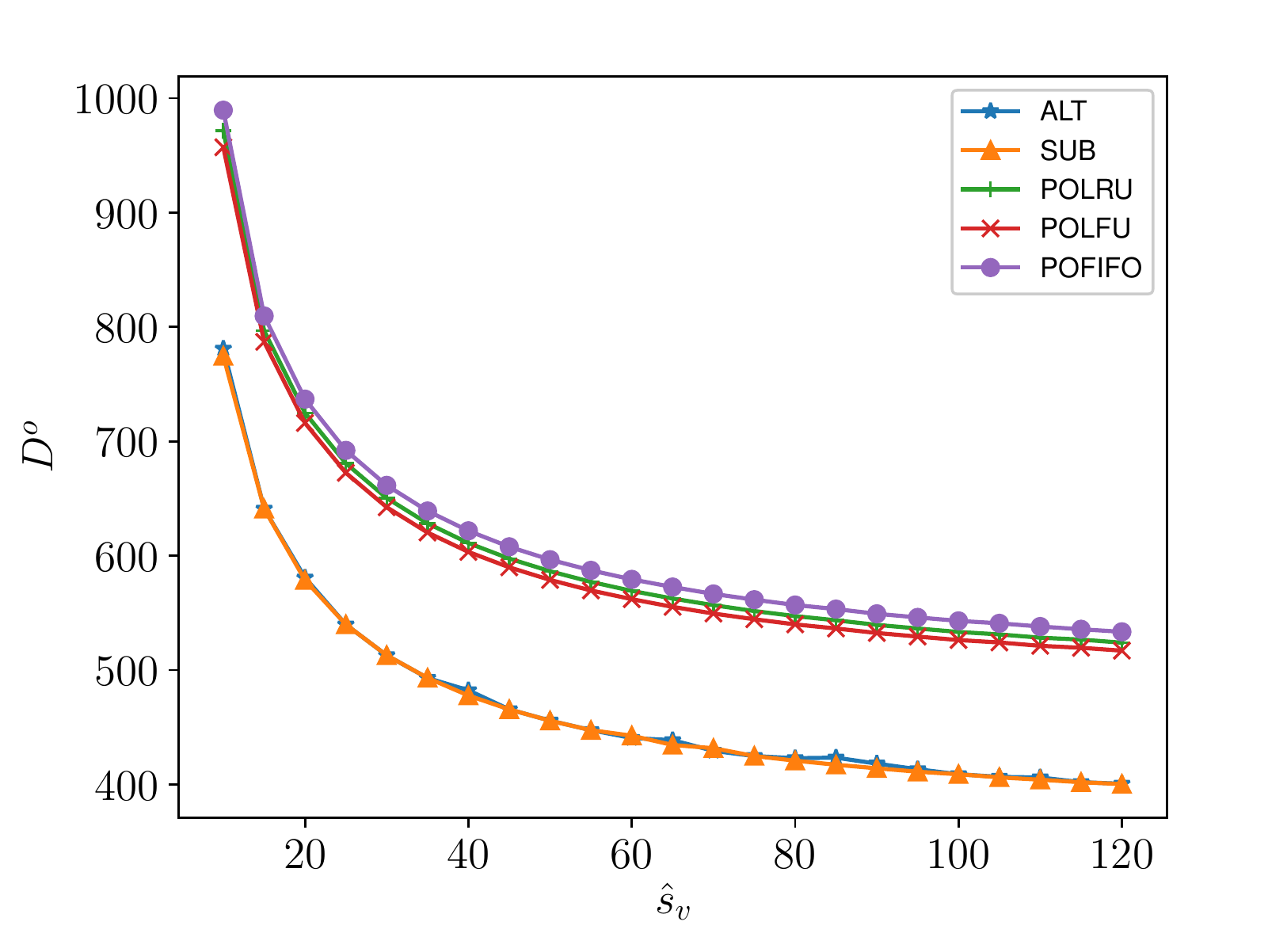}
    \vspace{-0.25cm}
    \caption{$D^o$ versus $\hat{s}_v$, which is the same for all $v \in V$. All SC cache capacities are equal to $c_{sc} = 2$, and $c_{mc} = 4$. $\gamma = 0.25$.}
    \label{fig:s2power}
\end{figure}

{\em Effect of cache capacity constraints.} We present the results of this setting in Fig. \ref{fig:s1cache}. We see that, with increasing cache capacities, our joint optimization algorithms reduce delay at a much faster rate 
compared to traditional replacement algorithms. SUB and ALT 
algorithms 
also achieve a point of minimum delay given large enough caches, while traditional algorithms do not converge to such a point and perform worse than SUB and ALT 
with all values of the cache capacity constraint. Numerically, SUB and ALT methods achieve at least 15\% less delay, with up to 50\% less delay at $c_{sc} = 4$, under the given parameters.

{\em Effect of power constraints.} We present the results of this setting in Fig. \ref{fig:s2power}. We observe that traditional methods and our algorithms show a similar decreasing trend in delay when the total power budget is increased. However, we can still observe the benefit of jointly optimizing power with caching: our algorithms achieve 25\% less delay compared to the best performing traditional method, POLFU.

{\em Convergence of SUB and ALT.} We present the results of this setting in Fig. \ref{fig:s3conv}. We observe that while both algorithms reach the minimum in similar times, ALT has a much steeper initial decrease in the relaxed delay-cost $D$. This is because the number of power variables is significantly smaller than the number of caching variables. ALT optimizes them separately which results in the initial steep decrease where power variables are being optimized whereas SUB optimizes them jointly leading to longer durations for each iteration.

\begin{figure}[t!]
    \centering
    \includegraphics[width=0.9\linewidth]{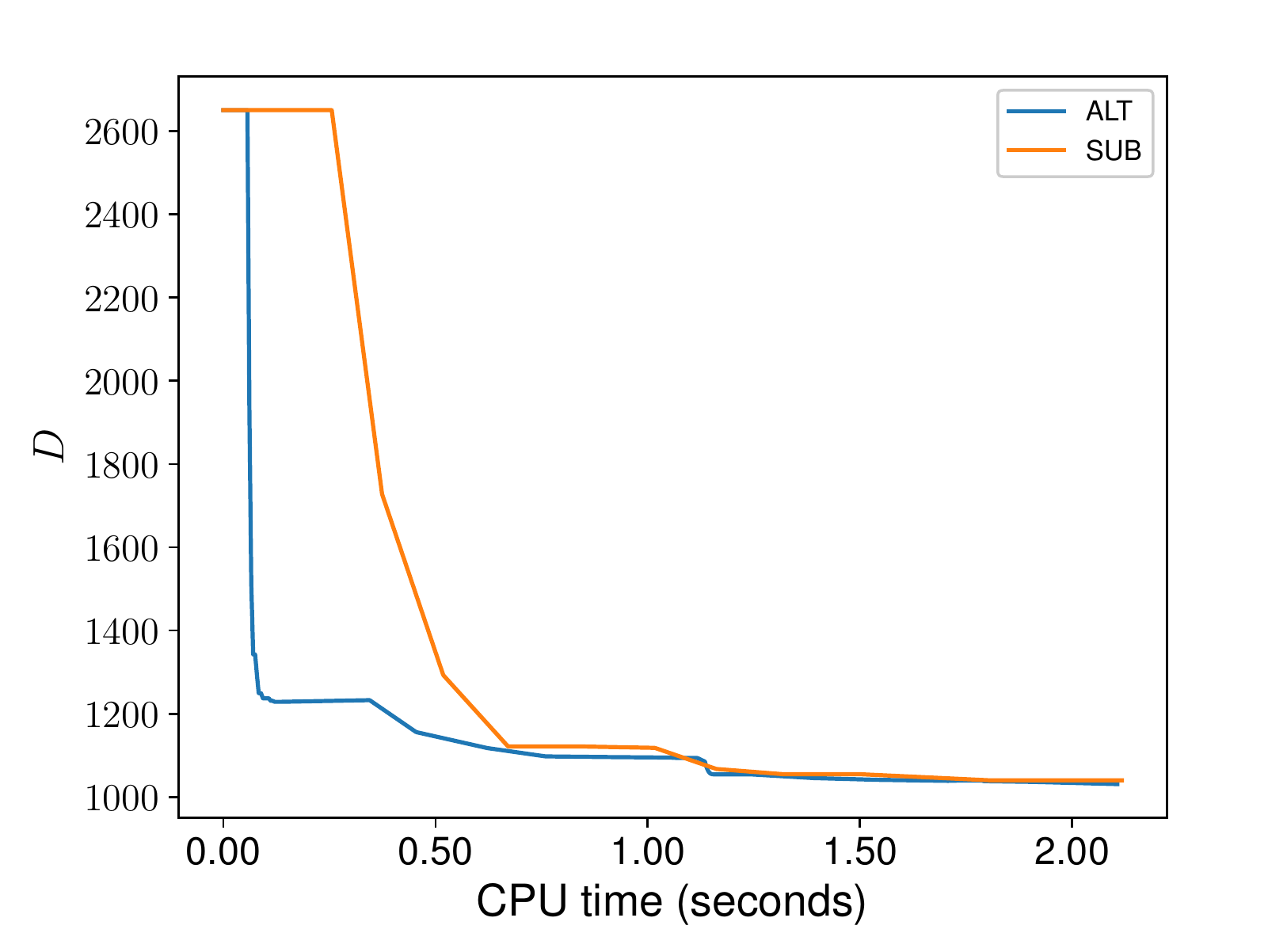}
    \vspace{-0.2cm}
    \caption{Convergence of SUB and ALT algorithms as described by $D$ with respect to time. $c_{sc} = 2$, $c_{mc} = 4$, $\hat{s}_v = 100$ and $\gamma = 0.25$.}
    \label{fig:s3conv}
\end{figure}

\section{Conclusion}
\label{conc} 
We considered the problem of joint power and caching optimization to minimize the transmission delay for a stationary request process in wireless HetNets. Because this problem is NP-complete, we studied several approximation methods that rely on convex relaxation and rounding of caching variables to construct an integral solution. 
More specifically, we provided necessary and sufficient conditions for the optimality of RCF to the joint optimization problem. We demonstrated Pareto optimality of the solution to RCF, and devised two solution techniques: alternating optimization technique when RCF satisfies biconvexity, and subgradient projection algoritm for general non-convex RCF. 
The results of our approach can enable the wireless HetNets to optimally exploit the resources to minimize the use of the backhaul connection, and hence to minimize the transmission delays in both mobile devices and the infrastructure, and to support latency-sensitive applications. They also quantify the potential cost savings from the deployment of SCs. More generally, optimal caching and power control algorithms represent a key enabling technology for realizing the potential of mobile edge computing and fog computing. 

Possible extensions of this work include devising decentralized techniques, and designing both uncoded and coded caching  via distributed adaptive stochastic descent algorithms. Furthermore, these results will inform the design of edge cloud architectures, by clarifying the relative benefits of centralized and distributed implementation.  
Querying for content can be seen as a simplified case of querying for a result of a computation or service. Thus, caching and routing algorithms are essential ingredients of an edge computing infrastructure which optimally schedules processing and job flows. Therefore, quantifying the potential cost savings from the deployment of SCs with caching capabilities via optimal routing algorithms is critical. Extensions also include a more detailed analysis of backhaul costs to effectively route the requests and control the traffic load on SCs and to overcome the transmission delay incurred in the backhaul due to limited bandwidth and dynamic channel conditions. 

\appendices

\section{Proof of Theorem \ref{optimalityguarantees}}\label{App:optimalityguarantees}
The proof follows from relaxing and bounding techniques. By Goemans and Williamson \cite{AgeSvi2004,GoeWil1994}, we have 
\begin{align}
\label{AgeevApprox}
\prod\limits_{l=1}^k (1-y_{p_l i})&\leq 1-(1-(1-1/k)^k)\min\Big\{1,\sum\limits_{l=1}^k y_{p_l i}\Big\}\nonumber\\
&\leq 1-(1-1/e)\min\Big\{1,\sum\limits_{l=1}^k y_{p_l i}\Big\},
\end{align}
as $(1-1/k)^k\leq 1/e$. On the other hand, we have $D^o(Y,S)=\mathbb{E}_{\nu}[D^o(X,S)]$. Using (\ref{AgeevApprox}), the 
relaxed cost function 
for serving a request (given that each request $(i, p)\in\mathcal{R}$ is well-routed) can be written concisely in terms of the allocation: 
\begin{align}
\label{Dlowerbound}
D(Y,S)
&\overset{(a)}{\geq}  \sum\limits_{(i,p)\in\mathcal{R}}{\lambda_{(i,p)}{\sum\limits_{k=1}^{|p|-1}f({\SINR}_{p_{k+1}p_k}(S)) \prod\limits_{l=1}^k (1-y_{p_l i}) }}\nonumber\\
&=D^o(Y,S),
\end{align}
where $(a)$ is due to (\ref{AgeevApprox}). We also upper bound $D(Y,S)$ as
\begin{align}
\label{Dupperbound}
D(Y,S)
&\overset{(b)}{\leq} \sum\limits_{(i,p)\in\mathcal{R}}{\lambda_{(i,p)}{\sum\limits_{k=1}^{|p|-1}f({\SINR}_{p_{k+1}p_k}(S)) \mathbb{E}_{\nu}\left[g_{p_k i}(X)\right] }},
\end{align}
where $(b)$ follows from the concavity of the $\min$ operator and employing (\ref{g_function}), and 
\begin{align}
&\mathbb{E}_{\nu}\left[g_{p_k i}(X)\right]\overset{(c)}{\leq} 1-(1-1/e)\mathbb{E}_{\nu}\Big[\min\Big\{1,\sum\limits_{l=1}^k x_{p_l i}\Big\}\Big] \nonumber\\
&=1-(1-1/e)\mathbb{E}_{\nu}\Big[1-\prod\limits_{l=1}^k x_{p_l i}\Big] \nonumber\\
&=1-(1-1/e)\Big(1-\prod\limits_{l=1}^k y_{p_l i}\Big), \nonumber
\end{align}
where $(c)$ is due to $(1-1/k)^k\leq 1/e$. Hence, $D^o(Y,S)\geq (D(Y,S)-D^{\rm{ub}}/e)/(1-1/e)$. Hence,
\begin{align}
D(Y,S)\leq
D^{\rm{ub}}/e+(1-1/e)D^o(Y,S),
\end{align}

From (\ref{Dlowerbound}) and (\ref{Dupperbound}), we have 
\begin{align}
\label{relaxedsolutionbounds}
\!\!\!\! D(Y,S) \geq D^o(Y,S)\geq (D(Y,S)-D^{\rm{ub}}/e)/(1-1/e).   
\end{align}

Because $Y^*\in\mathcal{D}_Y$ is optimal,  $D^o(Y^*,S)\leq D^o(Y^{**},S)$. From (\ref{relaxedsolutionbounds}) and the optimality of $Y^{**}$, 
$D^o(Y^{**},S)\leq D(Y^{**},S)\leq D(Y^{*},S)\leq D^{\rm{ub}}/e+(1-1/e)D^o(Y^{*},S)$.

\section{Proof of Proposition \ref{sufficient_condition_for_convexity}}\label{App:sufficient_condition_for_convexity}
Note that the objective function $D(S,Y)$ in (\ref{DasFunctionofYS}) is convex in caching variables $Y$. It is convex in $S$ if every $f({\SINR}_{p_{k+1}p_k}(S))$ is convex in $S$ where ${\SINR}_{p_{k+1}p_k}(S)$ is concave in $S$ for all $k$. However, given that ${\SINR}_{p_{k+1}p_k}$ is strictly increasing, $\nabla^2 {\SINR}_{p_{k+1}p_k}(S)$ cannot be negative definite. Letting $C=f^{-1}$, based on the observations \cite{huang2006distributed}, if
\begin{align}
\label{concavity_condition}
C^{''}(x)\cdot x+C^{'}(x)\leq 0,\quad \forall x\geq 0,    
\end{align} 
then $C$ is concave in $P\eqdef (\log(s_{vu}(i,p)))_{(v,u)\in E,\, (i,p):(u,v)\in p}$ (power measured in dB). From above relation, since $C^{''}(x)=\frac{2f^{'}(x)^2-f(x)f^{''}(x)}{f(x)^3}\leq 0$, and $f(x)\geq 0$, $\forall x\geq 0$, we have $2f^{'}(x)^2-f(x)f^{''}(x)\leq 0$, yielding $0\leq \frac{2f^{'}(x)^2}{f(x)}\leq f^{''}(x)$. Hence, $f({\SINR}_{p_{k+1}p_k})$ is convex in $P$.

The condition in (\ref{concavity_condition}) equivalently yields the following for $f^{-1}({\SINR}_{p_{k+1}p_k})$ to be concave in $P$:
\begin{align}
\label{concave_P_condition}
\frac{2f^{'}(x)^2-f(x)f^{''}(x)}{f(x)^3}\cdot x-\frac{f^{'}(x)}{f(x)^2}\leq 0.
\end{align}
Reordering the terms in (\ref{concave_P_condition}) we get the desired result. 

We note that in addition to the condition in (\ref{concave_P_condition}) on convexity in $\log$ power variables, $f({\SINR}_{vu}(S)$ might be convex in $s_{vu}$ when the nodes have a total power constraint as given by (\ref{SourcePowerConstraints}) which is satisfied with equality. We next explain the requirement under which the convexity in powers holds. Assuming that the total transmit power is fixed and equal to $\hat{s}$, since the routings are predetermined, a user's the total received power lumped with the noise power, coined $\bar{s}$, will be fixed given the allocation of the transmit power. Hence, we have $f({\SINR}_{vu}(S)=\frac{1}{\log_2(1+\frac{s_{vu}}{\bar{s}-s_{vu}})}$ for $(v,u)\in E$. Using this relation and computing the second order derivative of $f({\SINR}_{vu}(S)$ with respect to $s_{vu}$, it can be shown via algebraic manipulation that $f({\SINR}_{vu}(S)$ is convex in $s_{vu}$ provided that $s_{vu}\leq \frac{3}{4}\bar{s}$. We also emphasize that this condition is not restricted to the high $\SINR$ regime and is valid under general $\SINR$ values, provided that the total transmission power constraint is satisfied with equality.

\section{Proof of Proposition \ref{quasi_convex}}\label{App:quasi_convex}
With no loss of generality, assume that $S^a\neq S^b$ such that $D(Y,S^a)<D(Y,S^b)$. Then, a function $D$ defined on a convex subset $\mathcal{D}_Y\times \mathcal{D}_S$ of a real vector space is strictly quasi-convex in $S$ given $Y$, if for all $S^a\neq S^b$ and $\alpha \in (0,1)$ we have the following condition: 
\begin{align}
\label{strict_convex_D_in_S}
D(Y,\alpha S^a+(1-\alpha)S^b)<D(Y,S^b).    
\end{align}
Since $D$ is linear in $f$, it is easily verified that a sufficient condition for $D(Y,S)$ to be strictly quasi-convex in $S$ is when $f$ is strictly quasi-convex.
While the sum of quasiconvex functions defined on the same domain need not be quasiconvex, we will detail why quasi-convexity will be preserved in our setting.

For $f({\SINR}_{vu}(S^a))<f({\SINR}_{vu}(S^b))$ we have
\begin{align}
\label{f_quasi_convex}
&f({\SINR}_{vu}(\alpha S^a+(1-\alpha)S^b))\nonumber\\
&=\frac{1}{\log_2(1+{\SINR}_{vu}(\alpha S^a+(1-\alpha)S^b))}\nonumber\\
&<\frac{1}{\log_2(1+{\SINR}_{vu}(S^b))}
= f({\SINR}_{vu}(S^b)),
\end{align}
where we used (\ref{wireless_delay}). Under a total power constraint it is easy to note that $f({\SINR}_{vu}(S))$ decreases in $s_{vu}$. Since $f({\SINR}_{vu}(\alpha S^a+(1-\alpha)S^b))<\max\Big\{f({\SINR}_{vu}(S^a)),f({\SINR}_{vu}(S^b))\big\}=f({\SINR}_{vu}(S^b))$ for all $(v,u)\in E$ and $\max$ is Schur-convex, and by ordering the summands of $\sum\limits_{k=1}^{|p|-1}f({\SINR}_{p_{k+1}p_k}(S))$ in a decreasing manner, we have that
\begin{align}
\max\Big\{\sum\limits_{k=1}^{|p|-1}f({\SINR}_{p_{k+1}p_k}(S^a)),\sum\limits_{k=1}^{|p|-1}f({\SINR}_{p_{k+1}p_k}(S^b))\Big\}\nonumber\\
=\sum\limits_{k=1}^{|p|-1}f({\SINR}_{p_{k+1}p_k}(S^b)).\nonumber    
\end{align} 
Hence, we infer from (\ref{f_quasi_convex}) and the order-preserving mapping that (\ref{strict_convex_D_in_S}) holds which implies $D(Y,S^a)<D(Y,S^b)$.

We can also observe that
\begin{align}
&f(\alpha {\SINR}_{vu}(S^a)+(1-\alpha ){\SINR}_{vu}(S^b))\nonumber\\
&=\frac{1}{\log_2(1+\alpha {\SINR}_{vu}(S^a)+(1-\alpha ){\SINR}_{vu}(S^b))}\nonumber\\
&< \frac{1}{\log_2(1+{\SINR}_{vu}(S^b))}=f({\SINR}_{vu}(S^b)),
\end{align}
where inequality follows from that ${\SINR}_{vu}(S^a)> {\SINR}_{vu}(S^b)$ which implies $\alpha {\SINR}_{vu}(S^a)+(1-\alpha ){\SINR}_{vu}(S^b)> {\SINR}_{vu}(S^b)$. Furthermore, $f({\SINR}_{vu}(S))$ is a monotonically decreasing function of ${\SINR}_{vu}(S)$. Hence, $f({\SINR}_{vu}(S))$ is strictly quasi-convex in ${\SINR}_{vu}(S)$.

Note that $D(Y,S)$ is convex with respect to set $\mathcal{D}_Y$, which is due to (\ref{g_function}). 
Note also that $D$ is strictly quasi-convex in $Y$. This can be shown using the condition that if 
\begin{align}
D(Y^a,S)<D(Y^b,S),    
\end{align}
then for any $\lambda\in (0,1)$ it holds that
\begin{align}
D(\lambda Y^a+(1-\lambda)Y^b,S)<D(Y^b,S).    
\end{align}
To verify $D$ is strictly quasi-convex, it is necessary that $g_{p_k i}(Y)$'s given in (\ref{g_function}) are strictly quasi-convex. Assume $g_{p_k i}(Y^a)<g_{p_k i}(Y^b)$. Then for all $Y^a\neq Y^b$ and $\lambda\in (0,1)$:
\begin{align}
&g_{p_k i}(\lambda Y^a+(1-\lambda)Y^b)\nonumber\\
&=1-a_k\min\Big\{1,\sum\limits_{l=1}^k \lambda y_{p_l i}^a+(1-\lambda)y_{p_l i}^b\Big\} \nonumber\\
&\overset{(a)}{\leq} 1-a_k\Big[\lambda\min\Big\{1,\sum\limits_{l=1}^k  y_{p_l i}^a\Big\}
+(1-\lambda)\min\Big\{1,\sum\limits_{l=1}^k y_{p_l i}^b\Big\}\Big]\nonumber\\
&=\lambda g_{p_k i}(Y^a)+(1-\lambda)g_{p_k i}(Y^b)
<g_{p_k i}(Y^b),
\end{align}
where $(a)$ is due to the concavity of the $\min$ function. This verifies that $D$ is strictly quasi-convex in $Y$. 

\section{Proof of Theorem \ref{necessary_optimal}}\label{App:necessary}


If the optimization problem satisfies some regularity conditions \cite{bertsekas1998nonlinear}, the necessary conditions for a solution of the nonlinear RCF in (\ref{mindelay_nonlinearproblem}) are given by the KKT conditions. Then, the following four groups of conditions hold:

(i) For stationary, the local solution $(Y^{**},\,S^{**})$ needs to satisfy the following subgradients with respect to $Y=(y_{vi})$ and the gradients with respect to $S=(s_{vu})$ values:
\begin{align}
d_{y_{vi}} D(Y^{**},S)+ \sum\limits_{v\in V,\, i\in\mathcal{C}} \mu_{vi}-\sum_{v\in V,\, i\in\mathcal{C}} &\nu_{vi} \nonumber\\
+ \sum_{v\in V}\eta_v &=0, \,\, v\in V,\, i\in\mathcal{C},\nonumber\\
\frac{\partial D(Y,S^{**})}{\partial s_{vu}}+ \sum\limits_{v\in V} \beta_v-\sum_{e\in E,\, r\in\mathcal{R}} \gamma_{er}&=0, \nonumber\\
(u,v)\in V,\,(i,p)\in\mathcal{R}.\nonumber
\end{align}

(ii) For primal feasibility, we require that 
\begin{align}
\label{Primal1}
y^{**}_{vi}-1&\leq 0,\quad v\in V,\,\, i\in \mathcal{C},\\
\label{Primal2}
-y^{**}_{vi}&\leq 0,\quad v\in V,\,\, i\in \mathcal{C},\\
\label{Primal3}
\sum\limits_{i\in \mathcal{C}}{y_{vi}} &\leq c_v,\,\,\, v\in V,\\
\label{Primal4}
\sum\limits_{u\in O_v}  s^{**}_{vu}-\hat{s}_v&\leq 0,\quad v\in V,\\
\label{Primal5}
-s^{**}_{vu}&\leq 0,\quad v,\,u\in V.
\end{align}
For dual feasibility, we require that
\begin{align}
    \mu_{vi},\,\, \nu_{vi},\,\, \eta_{v},\,\, \beta_v&\geq 0, \,\, \quad v\in V,\,\ i\in\mathcal{C},\\ 
    \gamma_{e,r}&\geq 0,\quad e=(u,v)\in E, \,\, r=(i,p)\in\mathcal{R}.\nonumber
\end{align}

(iii) The complementary slackness conditions are given as
\begin{align}
\label{slackness}
\mu_{vi}(y^{**}_{vi}-1)&=0,\quad v\in V,\,\, i\in \mathcal{C},\nonumber\\
\nu_{vi}\cdot y^{**}_{vi}&=0,\quad v\in V,\,\, i\in \mathcal{C},\nonumber\\
\eta_v\Big(\sum\limits_{i\in \mathcal{C}}{y_{vi}} - c_v\Big)&=0,\quad v\in V,\nonumber\\
\beta_v\Big(\sum\limits_{u\in O_v}  s^{**}_{vu}-\hat{s}_v\Big)&=0,\quad v\in V,\nonumber\\
\gamma_{e,r}\cdot s^{**}_{vu}&=0,\quad e=(u,v)\in E,\,\, r\in\mathcal{R}.
\end{align}

(iv) The subgradients with respect to $y_{vi}$ should satisfy
\begin{align}
\label{gradients_caching}
&d_{y_{vi}} D(Y^{**},\, S)\nonumber\\
&=\sum\limits_{u\in O_v} \sum\limits_{(i,p):(u,v)\in p}  \!\!\!\!\lambda_{(i,p)} \sum\limits_{k=1}^{|p|-1} \!\!f({\SINR}_{p_{k+1}p_k}(S)) d_{y_{vi}} g_{p_k i}(Y^{**})\nonumber\\
&=-\mu_{vi}+\nu_{vi}+\eta_v=\alpha_{vi},\quad v\in V,\,\, i\in\mathcal{C}. 
\end{align}
When $y_{vi}=0$, constraint (\ref{Primal1}) is eliminated, and $d_{y_{vi}} D(Y^{**},\, S)\geq \alpha_{vi}$. Similarly, if $y_{vi}=1$, constraint (\ref{Primal2}) is eliminated, and $d_{y_{vi}} D(Y^{**},\, S)< \alpha_{vi}$. This verifies (\ref{caching_subgradients}).

(v) The gradients with respect to the power variables are 
\begin{align}
\label{gradients_power}
&\frac{\partial D(Y,\, S^{**})}{\partial s_{vu}}\nonumber\\
&=\sum\limits_{u\in O_v} \sum\limits_{(i,p):(u,v)\in p} \lambda_{(i,p)} \frac{\partial f({\SINR}_{v u}(S^{**}))}{\partial s_{vu}} g_{ui}(Y)\nonumber\\
&=-\beta_v+\gamma_{e,r},\quad e=(u,v)\in E, \,\, r=(i,p)\in\mathcal{R}.
\end{align}
Solving the gradients (\ref{gradients_caching}) and (\ref{gradients_power}), along with the complementary slackness conditions in (\ref{slackness}), we obtain the necessary conditions, which concludes the proof.

\section{Proof of Theorem \ref{pareto}}\label{App:pareto}
The proof follows from employing similar techniques as in \cite[Sect. VI-C, Theorem 4]{XiYeh2008}.

We initially assume that the joint power and cache allocation problem is strictly quasi-convex. With this assumption, for a fixed cache allocation $(y_{vi}^{**})_{v\in V,\, i\in\mathcal{C}}$,  the relaxed delay-cost is a convex function of $S$. Therefore, any feasible power allocation $S^*$ satisfying (\ref{necessary_s_constraint}) satisfies that
\begin{align}
\label{feasibleS}
D(Y,S^{**})=\underset{S\in\mathcal{D}_S}{\min} \sum\limits_{(i,p)\in\mathcal{R}} D_{(i,p)}(Y,S).    
\end{align}
Given any feasible power allocation, if condition (\ref{necessary_y_constraint}) holds at cache allocation $(y_{vi}^{**})_{v\in V,\, i\in\mathcal{C}}$, then
\begin{align}
\label{feasibleY}
D(Y^{**},S)=\underset{Y\in\mathcal{D}_Y}{\min} \sum\limits_{(i,p)\in\mathcal{R}} D_{(i,p)}(Y,S).    
\end{align}

In this case, any initial power and cache allocation configuration can be driven to a limiting $(Y^{**},\,S^{**})$ such that the condition (\ref{necessary_s_constraint}) is satisfied at $S^{**}$ given $Y^{**}$, and $Y^{**}$ satisfies (\ref{necessary_y_constraint}) given $S^{**}$. We now suppose that under the more general 
convex power allocation region model, 
there are algorithms that also can drive the power and cache configuration to a limit $(Y^{**},\,S^{**})$ such that the conditions (\ref{necessary_s_constraint}) and (\ref{necessary_y_constraint}) hold simultaneously. Although global optimality cannot be guaranteed, the Pareto optimality can be shown.

Suppose that $D(Y^{\#},\, S^{\#})$ Pareto dominates $D(Y^{**},S^{**})$. Without loss of generality, we can assume
\begin{align}
D_{(m,r)}(Y^{\#},\, S^{\#}) < D_{(m,r)}(Y^{**},S^{**}).
\end{align}
Because both $Y^{\#}$ and $Y^{**}$ belong to $\mathcal{D}_Y$, and $\mathcal{D}_Y$ is strictly convex, $Y^{\beta}=\beta Y^{**}+(1-\beta)Y^{\#}$ is achievable for all $\beta\in(0,1)$. Moreover, $Y^{\#}\neq Y^{**}$ because otherwise, $S^{\#}\neq S^{**}$, and from Pareto domination we would have
\begin{align}
\sum\limits_{(i,p)\in\mathcal{R}} D_{(i,p)}(Y^{**},\, S^{\#}) < \sum\limits_{(i,p)\in\mathcal{R}} D_{(i,p)}(Y^{**},S^{**}).
\end{align}
However, this contradicts (\ref{feasibleS}). Therefore, $Y^{\beta}$ is in the interior of $\mathcal{D}_Y$ for any $\beta\in (0,1)$. 

From the same reasoning, $S^{\#}\neq S^{**}$ and $S^{\alpha}=\alpha S^{**}+(1-\alpha)S^{\#}$ is feasible for any $\alpha\in [0,1]$ simply by linearity of feasible power allocations. 

Since $D_{(i,p)}$ is strictly quasi-convex, $D(Y^{\alpha},S^{\alpha})$ Pareto dominates $D(Y^{**},S^{**})$ as well for any $\alpha\in (0,1)$, since $D_{(m,r)}(Y^{\alpha},S^{\alpha})<D_{(m,r)}(Y^{**},S^{**})$, and $D_{(i,p)}(Y^{\alpha},S^{\alpha})\leq D_{(i,p)}(Y^{**},S^{**})$ for $(i,p)\neq (m,r)$. Summing up all the terms on LHS and RHS, we have for any $\alpha\in (0,1)$
\begin{align}
\label{D_alpha_vs_D_star}
    \sum\limits_{(i,p)\in\mathcal{R}} D_{(i,p)}(Y^{\alpha},S^{\alpha})<\sum\limits_{(i,p)\in\mathcal{R}} D_{(i,p)}(Y^{**},S^{**}).
\end{align}

By optimality condition (\ref{necessary_s_constraint}) and the fact that $Y^{\alpha}$ is in the interior of $\mathcal{D}_Y$ for any $\alpha\in (0,1)$, we have for any $\alpha\in (0,1)$ and $(v,u)\in E$,
\begin{align}
&\sum\limits_{(i,p)\in\mathcal{R}}\frac{\partial D_{(i,p)}}{\partial s_{vu}}(Y^{**},S^{**}) (S^{\#}-S^{**})\nonumber\\
&\overset{(a)}{=}\frac{1}{1-\alpha}\sum\limits_{(i,p)\in\mathcal{R}}\frac{\partial D_{(i,p)}}{\partial s_{vu}}(Y^{**},S^{**}) (S^{\alpha}-S^{**})\nonumber\\
&>\frac{1}{1-\alpha}\sum\limits_{(i,p)\in\mathcal{R}}\frac{\partial D_{(i,p)}}{\partial s_{vu}}(Y^{**},S^{**}) (\bar{S}^{\alpha}-S^{**})\geq 0,\nonumber
\end{align}
where $(a)$ follows from $S^{\alpha}=\alpha S^{**} + (1-\alpha) S^{\#}$, and $\bar{S}^{\alpha}$ is some power matrix strictly dominating $S^{\alpha}$. Following similar steps, from (\ref{necessary_s_interference_constraint}) for $j\in I_u,\,\, (i',p'): (u,j)\in p'$   
\begin{align}
&\sum\limits_{(i,p)\in\mathcal{R}} \frac{\partial D_{(i,p)}}{\partial s_{ju}}(Y^{**},S^{**})(S^{\#}-S^{**})\nonumber\\
&>\frac{1}{1-\alpha}\sum\limits_{(i,p)\in\mathcal{R}}\frac{\partial D_{(i,p)}}{\partial s_{ju}}(Y^{**},S^{**}) (\bar{S}^{\alpha}-S^{**})\geq 0.\nonumber
\end{align}

Since $D_{(i,p)}$ is twice continuously differentiable on $S$, there exists $\epsilon>0$ such that for all $\alpha\in [1-\epsilon,1)$
\begin{align}
\label{S_alpha_inequality}
\sum\limits_{(i,p)\in\mathcal{R}}\frac{\partial D_{(i,p)}}{\partial s_{vu}}(Y^{\alpha},S^{**}) (S^{\alpha}-S^{**}) &\geq 0, \\
\label{S_alpha_interference_inequality}
\sum\limits_{(i,p)\in\mathcal{R}} \frac{\partial D_{(i,p)}}{\partial s_{ju}}(Y^{\alpha},S^{**})(S^{\alpha}-S^{**})&\geq 0,\,\, \nonumber\\
j\in I_u,\,\, (i',p'): (u,j)\in p'.
\end{align}
Combining (\ref{S_alpha_inequality}) with the convexity of $D_{(i,p)}(Y^{\alpha},\cdot)$ implies
\begin{align}
    \sum\limits_{(i,p)\in\mathcal{R}} D_{(i,p)}(Y^{\alpha}, S^{**})&\leq \sum\limits_{(i,p)\in\mathcal{R}} D_{(i,p)}(Y^{\alpha}, S^{\alpha}) \nonumber\\
    &<\sum\limits_{(i,p)\in\mathcal{R}} D_{(i,p)}(Y^{**}, S^{**})
\end{align}
where the second inequality comes from (\ref{D_alpha_vs_D_star}). However, this result contradicts (\ref{feasibleY}). Hence, there does not exist another pair of feasible allocations $((y_{vi}^{\#}),\, (s^{\#}_{vu}))$ such that $D_{(i,p)}(Y^{\#},S^{\#})\leq D_{(i,p)}(Y^{**},S^{**})$, $\forall (i,p)\in\mathcal{R}$, with at least one inequality being strict.

\section{Proof of Lemma \ref{LemmaConvergence}} \label{App:LemmaConvergence}
For convenience, we just prove for the case where the stepsize is jointly computed as
\begin{align*}
    \xi^t = \frac{D(\boldsymbol{y}^{t},S^t) - \hat{D}^t}{\norm{d_S^t}^2 + \norm{d_{\boldsymbol{y}}^t}^2} 
\end{align*}

Let $D^{*}$ to be the local minima attained with $(\boldsymbol{y}^*_{sub},S^*_{sub})$ that satisfies the KKT condition. To obtain a contradiction, we assume that there exist $\epsilon > 0$ such that
\begin{equation*}
    \liminf_{t \to \infty} D(\boldsymbol{y}^{t},S^t) > D^*_{sub} + \epsilon. 
\end{equation*}
Then by the continuity of $D(\boldsymbol{y},S)$, there exist $(\boldsymbol{\hat{y}},\hat{S})$ near the local minima such that $D(\boldsymbol{\hat{y}},\hat{S}) = D^*_{sub} + \epsilon$. Therefore, there exist a subsequence $\{l\}$ of $\{t\}$ such that $D(\boldsymbol{y}^l,S^l) \geq D(\boldsymbol{\hat{y}},\hat{S}) + \delta_l$. Thus, we obtain
\begin{equation*}
    \hat{D}^l = \min_{0 \leq j \leq l} - \delta_l \geq D(\boldsymbol{\hat{y}},\hat{S})
\end{equation*}
for any $l \geq l_0$ where $l_0$ is some positive integer, i.e., it holds that
\begin{equation*}
    D(\boldsymbol{y}^l,S^l) - D(\boldsymbol{\hat{y}},\hat{S}) \geq D(\boldsymbol{y}^l,S^l) - \hat{D}^l.
\end{equation*}

Furthermore, for a subgradient projection method with any stepsize, we have
\begin{equation*}
    \norm{x^{t+1} - \hat{x}}^2 \leq \norm{x^t - \hat{x}}^2 - 2\xi^t\left(f(x^t) - f(\hat{x})\right) + (\xi^t)^2\norm{d^t}^2,
\end{equation*}
where $\hat{x}$ is any available variable. Therefore, it holds that
\begin{align*}
    &(\norm{\boldsymbol{y}^{l+1} - \boldsymbol{\hat{y}}}^2 + \norm{S^{l+1} - \hat{S}}^2) 
    \\ &\leq (\norm{\boldsymbol{y}^{l} - \boldsymbol{\hat{y}}}^2 + \norm{S^{l} - \hat{S}}^2) 
    \\ &- 2\xi^l\left(D(\boldsymbol{y}^l,S^l) - \hat{D}^l\right) + (\xi^l)^2\left(\norm{d^l_{\boldsymbol{y}}}^2 + \norm{d^l_S}^2\right).
\end{align*}

Since we employ the modified Polyak's stepsize, we have
\begin{align*}
    (\norm{\boldsymbol{y}^{l+1} - \boldsymbol{\hat{y}}}^2 + \norm{S^{l+1} - \hat{S}}^2) 
    &\leq 
     (\norm{\boldsymbol{y}^{l} - \boldsymbol{\hat{y}}}^2 + \norm{S^{l} - \hat{S}}^2) \\&
    - \frac{D(\boldsymbol{y}^l,S^l) - \hat{D}^l}{\norm{d^l_{\boldsymbol{y}}}^2 + \norm{d^l_S}^2}.
\end{align*}
Using a telescoping sum, the previous bound yields the following inequality:
\begin{align*}
    &(\norm{\boldsymbol{y}^{l+1} - \boldsymbol{y}^{*}_{sub}}^2 + \norm{S^{l+1} - S^{*}_{sub}}^2) 
    \\&\leq
     (\norm{\boldsymbol{y}^{0} - \boldsymbol{y}^{*}_{sub}}^2 + \norm{S^{0} - S^{*}_{sub}}^2) 
     \\ & - \sum_{m = 0}^{l}\frac{D(\boldsymbol{y}^m,S^m) 
     - \hat{D}^m}{\norm{d^m_{\boldsymbol{y}}}^2 + \norm{d^m_S}^2}.
\end{align*}
As a result, in the limit as $l$ goes to infinity, it holds that 
\begin{align}\label{acumsum}
    \lim_{t \to \infty} \sum_{m = 0}^{t}\frac{D(\boldsymbol{y}^m,S^m) - \hat{D}^m}{\norm{d^m_{\boldsymbol{y}}}^2 + \norm{d^m_S}^2}  
    &\leq (\norm{\boldsymbol{y}^{0} - \boldsymbol{y}^{*}_{sub}}^2 
    \nonumber\\&+ \norm{S^{0} - S^{*}_{sub}}^2) < \infty.
\end{align}

On the other hand, assume that the subgradient norm $\norm{d^m_{\boldsymbol{y}}}^2 + \norm{d^m_S}^2$ is upper bounded by some positive scalar $U$ in $m$-th step for any $m$. Then we obtain
\begin{align*}
    \lim_{t \to \infty} \sum_{m = 0}^{t}\frac{D(\boldsymbol{y}^m,S^m) - \hat{D}^m}{\norm{d^m_{\boldsymbol{y}}}^2 + \norm{d^m_S}^2} \geq \sum_{m = 0}^t \frac{\delta}{U^2} = \infty,
\end{align*}
which contradicts with (\ref{acumsum}). Thus the subgradient projection algorithm with modified Polyak's stepsize converges to $D^*_{sub}$.

\section{Proof of Lemma \ref{LemmaConvergenceRate}} \label{App:LemmaConvergenceRate}
We first show that (\ref{SharpSetofMinima}) holds in our case. 

By Lemma \ref{LemmaConvergence}, with diminishing $\delta_t$, the algorithm is guaranteed to converge to local minima $D^*_{sub}$. Also, since the constraint set is finite, there exist a positive integer $t_C$, such that $(\boldsymbol{y}^t,S^t), \forall t \geq t_C$ are in a subset $(\mathcal{D}^C_{Y},\mathcal{D}^C_S) \subseteq (\mathcal{D}_{Y},\mathcal{D}_S)$ where the objective function $D(\boldsymbol{y},S)$ is minimized on $(\mathcal{D}^C_{Y},\mathcal{D}^C_S)$ by any variable in the set $\left(\boldsymbol{y}^*_{sub},S^*_{sub}\right)$, i.e. $D^*_{sub} = \min_{(\boldsymbol{y},S) \in (\mathcal{D}^C_{Y},\mathcal{D}^C_S)} D(\boldsymbol{y},S)$, with the corresponding minima set $\left(\boldsymbol{y}^*_{sub},S^*_{sub}\right)$ be compact and simply connected.

In (\ref{SharpSetofMinima}), if $D(\boldsymbol{y},S) = D^*_{sub}$, then $\left(\boldsymbol{y},S\right) \in \left(\boldsymbol{y}^*_{sub},S^*_{sub}\right)$, obviously (\ref{SharpSetofMinima}) holds with both side equal to $0$. If $D(\boldsymbol{y},S) > D^*_{sub}$, then consider the shortest ascending path from any point in the set $(\boldsymbol{y}^*_{sub},S^*_{sub})$ to $(\boldsymbol{y},S)$, by the generalization of mean value theorem to the subgradient case 
 (see Theorem 4.2 in \cite{studniarski1985mean}) we have
\begin{equation*}
    D(\boldsymbol{y},S) - D^*_{sub} \geq \norm{d_{min}} L(\boldsymbol{y},S)
\end{equation*}
where $L$ is given in (\ref{L_function}) and
\begin{align}
d_{min} = \min_{\substack{ (\boldsymbol{y}^\prime, S^\prime) \in \\ (\mathcal{D}^C_{Y},\mathcal{D}^C_S) \setminus (\boldsymbol{y}^*_{sub},S^*_{sub})}} \sqrt{\norm{d_{\boldsymbol{y}}}^2 + \norm{d_S}^2}    \nonumber
\end{align}
is the minimum subgradient that is not in the minima set.

Meanwhile, since the set $(\mathcal{D}^C_{Y},\mathcal{D}^C_S)$ can be chosen finite and compact with the gradient on variable $S$ not being arbitrarily close to $\textbf{0}$ or to infinity. Let $U \geq \norm{d_S} \geq \mu$ with some positive scalar $U$ and $\mu$, then $d_{min} \geq \mu$, thus we have (\ref{SharpSetofMinima}) holds within the set $(\mathcal{D}^C_{Y},\mathcal{D}^C_S)$.

Then we show (\ref{convergelinearly}). As in the proof of Lemma \ref{LemmaConvergence}, we have
\begin{align*}
    &\norm{\boldsymbol{y}^{t+1} - y}^2 + \norm{S^{t+1} - s}^2 
    \\&\leq 
     \norm{\boldsymbol{y}^{t} - y}^2 + \norm{S^{t} - s}^2 - \frac{D(\boldsymbol{y}^t,S^t) - D^*_{sub}}{\norm{d^t_{\boldsymbol{y}}}^2 + \norm{d^t_S}^2}
\end{align*}
for any $(y,s) \in (\boldsymbol{y}^*_{sub},S^*_{sub})$ and any $t \geq t_C$, which leads to
\begin{align*}
    L(\boldsymbol{y}^{t+1},S^{t+1})^2 \leq L(\boldsymbol{y}^t,S^t)^2 
    - \frac{D(\boldsymbol{y}^t,S^t) - D^*_{sub}}{U^2}.
\end{align*}

Since (\ref{SharpSetofMinima}) holds, we further have
\begin{align*}
    L(\boldsymbol{y}^{t+1},S^{t+1})^2 &\leq L(\boldsymbol{y}^t,S^t)^2 
    - \frac{\mu^2}{U^2}L(\boldsymbol{y}^t,S^t)^2
    \\& = \left(1 - \frac{\mu^2}{U^2}\right)L(\boldsymbol{y}^t,S^t)^2.
\end{align*}
Thus, (\ref{convergelinearly}) holds, the algorithm converges linearly after the $t_C$-th step where $t_C$ is a finite number. 


\bibliographystyle{IEEEtran}
\bibliography{JointRoutingCachingPowerOptimizationReferences}

\end{document}